\documentclass[12pt, draftclsnofoot, onecolumn]{IEEEtran}
\usepackage{times,amsmath,epsfig,latexsym,amssymb,psfrag,graphicx}
\usepackage{paralist}
\usepackage{booktabs}
\usepackage{color}
 \usepackage{setspace}
\usepackage{epstopdf}
\usepackage{subcaption}
\usepackage[ruled,vlined]{algorithm2e}
\usepackage{citesort}
\newtheorem{proposition} {Proposition}
\newtheorem{definition} {Definition}

\newenvironment{proof}[1][Proof]{\begin{trivlist}
\item[\hskip \labelsep {\bfseries #1}]}{\end{trivlist}}

\newcommand{\qed}{\nobreak \ifvmode \relax \else
\ifdim\lastskip<1.5em \hskip-\lastskip
\hskip1.5em plus0em minus0.5em \fi \nobreak
\vrule height0.75em width0.5em depth0.25em\fi}
\begin{document}

\title{$E^2$-MAC: Energy Efficient Medium Access  for Massive M2M Communications\footnote{Accepted in IEEE Transactions on communications. © 2016 IEEE. Personal use of this material is permitted. Permission from IEEE must be obtained for all other
uses, including reprinting/republishing this material for advertising or promotional purposes, collecting new
collected works for resale or redistribution to servers or lists, or reuse of any copyrighted component of this
work in other works.
}}
\author{Guowang Miao$^*$, Amin Azari$^*$,  and Taewon Hwang$^+$\\
$^*$KTH Royal Institute of Technology, $^+$Yonsei University\\
Email: \{guowang, aazari\}@kth.se, twhwang@yonsei.ac.kr}
\maketitle


\begin{abstract}

In this paper, we investigate energy-efficient clustering and medium access control (MAC)  for cellular-based M2M networks to minimize device energy consumption and prolong network battery lifetime. First, we present an accurate energy consumption model that considers both static and dynamic energy consumptions, and utilize this model to derive the network lifetime. Second, we find the cluster size to maximize the network lifetime and develop an energy-efficient cluster-head selection scheme. Furthermore, we find feasible regions where clustering is beneficial in enhancing  network lifetime. We further investigate communications protocols for both intra- and inter-cluster communications. While   inter-cluster communications use conventional cellular access schemes, we develop an energy-efficient and load-adaptive multiple access scheme, called $n$-phase CSMA/CA, which provides a tunable tradeoff between energy efficiency, delay, and spectral efficiency of the network.   The simulation results show that the proposed clustering, cluster-head selection, and communications protocol design outperform the others in energy saving and significantly prolong the lifetimes of both individual nodes and the whole M2M network.

\end{abstract}
\begin{IEEEkeywords}
Machine to Machine communications, Internet of Things, MAC, Energy efficiency, Lifetime, Delay.
\end{IEEEkeywords}

\IEEEpeerreviewmaketitle

 \section{Introduction}\label{intr}
\IEEEPARstart{I}{nternet} of Things (IoT) enables smart devices to participate more actively in everyday life, business, industry, and health care. Among large-scale applications, cheap and widely spread machine-to-machine (M2M) communications supported by cellular networks will be one of the most important enablers for the success of IoT \cite{itu}. M2M communications, also known as machine-type communications (MTC), means the communications of machine devices without human intervention \cite{ref2}. 
The characteristics of MTC are:  small
packet payload, periodic or event-driven traffic, extremely high
node density, limited power supply, limited computational
capacity, and limited radio front-ends.
 Also, smart devices are usually battery-driven and long battery life is crucial for them, especially for devices in remote areas, as there would be a huge amount of maintenance effort if their battery lives are short. Based on the 5G envision from Nokia \cite{nok}, the bit-per-joule energy efficiency for cellular-based machine-type communications must be improved
by a factor of ten in order to provide 10 years of battery lifetimes.   

\subsection{Literature study} 
The lifetime issue in M2M networks is similar to that in wireless sensor networks (WSNs). In the following, we briefly introduce state-of-the-art medium access control (MAC) and clustering design for both wireless sensor networks and cellular networks. 

\subsubsection{MAC and clustering design in WSNs}\label{mw}
Wireless sensor networks play an important role in many industrial, monitoring, health-care, and military applications.  The evolution of MAC protocols for WSNs is investigated in \cite{evo}. The evolution of clustering algorithms for WSNs is investigated in \cite{afnc}, which classifies the available clustering algorithms  depending on cluster formation criteria and parameters used for cluster-head (CH) selection. 
Along with the proposed MAC and clustering protocols in literature, some standardization efforts have been done like IEEE 802.15.4 and WirelessHART. 
MAC design for wireless sensors over cellular networks is investigated in recent years.  In \cite{cog}, sensor nodes form local area networks and communicate with data-gathering node(s) through gateways and base stations (BSs). In \cite{converge}, a model for WSN and LTE-advanced network convergence is proposed. 
The literature study shows that while energy efficiency has been a key factor in WSN design, an overly simplified energy consumption model has been used  in these  WSN research works which 
 usually assumes fixed energy consumption in each operating modes. This assumption no longer works in cellular networks as transmission energy may vary significantly to compensate path loss and is comparable or even much larger than  circuit energy consumption. Furthermore, direct application of WSN MAC designs in cellular-based M2M networks is either inefficient or impossible because: (i) cellular-based M2M networks have unique characteristics, e.g. massive concurrent access requests and diverse quality of service (QoS) requirements for machine nodes, which are quite different from WSNs; and (ii) the existence of BSs in cellular networks enables network assistance to improve device energy efficiency which is rarely considered in WSN literature. Then, the existing MAC and clustering protocols for WSNs fail to enable M2M communications in cellular networks \cite{scal2}.

\subsubsection{MAC design in cellular-based M2M networks}\label{rsec}
Random access channel (RACH) of the LTE-Advanced is the typical way for machine nodes to access the base station \cite{andres}.  
 The capacity limits of RACH for serving M2M communications and a survey of improved alternatives are studied in \cite{laya}.
  Among the alternatives, access class barring (ACB) is a promising approach which has attracted lots of attentions in literature \cite{acb1}. 
In \cite{scal2}, it is proposed to divide each communications frame into two periods: one for contention and the other for data transmission. 
The proposed schemes in \cite{scal2} and \cite{acb1} save energy by preventing collisions in data transmission. However, they require machine nodes to be active for a long time to gain channel access, which is not energy efficient.
A time-controlled access framework  satisfying the delay requirements of a massive M2M network is proposed in \cite{mas}, where the authors propose to divide machine nodes into  classes based on   QoS requirements and fixed access intervals are provided for each class. 
Power-efficient MAC protocols for machine devices with reliability constraints are considered in \cite{poshti}. The energy-efficient scheduling of machine devices in LTE networks together with  cellular users is 
investigated in \cite{coex}. While the energy-efficient solutions in \cite{poshti}-\cite{coex} are useful for  direct communications between machine devices and the BS, enabling large-scale M2M communications over cellular networks requires an energy efficient MAC protocol which  tackles also the massive concurrent access issues. The energy-efficient massive concurrent access control to the shared wireless medium is still an open problem for massive M2M communications  and is investigated in this work.
\subsubsection{Clustering design in cellular-based M2M networks}
Feasibility of clustering for machine-type devices in cellular networks has been investigated in \cite{gsip} to address the massive access-request problem. In \cite{mc1}, given the initial set of CHs, each machine node is connected to its nearest cluster and in each cluster, the node with the lowest communication cost is selected as the CH. In \cite{zki}, the outage-optimized density of data collectors in a capillary network, where the machine  devices  and  data  collectors  are  randomly  deployed within  a  cell, is derived. An emerging communication paradigm in cellular networks is direct Device-to-Device (D2D) communications \cite{dzor}.  D2D communications motivates the idea to  aggregate and relay M2M traffic through D2D links \cite{d2dagg}.   Without an  installed  gateway, each machine node could act as a CH \cite{grlaya}. The study of clustered M2M communications with battery-limited nodes as the CHs  is absent in literature and is the focus of this paper. Also, the existence of BSs in cellular networks enables network assistance to further improve clustering performance, which has not been considered in literature and we will take this into account as well.
\subsection{Open problems and Contributions}\label{ope}
As discussed above, there are promising MAC and clustering protocols in WSN literature and standardizations. However, considering the particular characteristics of cellular-based M2M communications, direct applications of these protocols in cellular-based M2M networks is either impossible or inefficient. Moreover, the energy consumption model in these works is overly simplified.  Addressing the numerous concurrent machine access within the current cellular network infrastructure  in an energy-efficient way is still an open problem and is the  focus of this paper. The main contributions of this paper include:

\begin{itemize}
\item
Present a lifetime-aware MAC design framework.
Use an accurate energy consumption model by taking both transmission and circuit energy consumptions into account. 

\item
Explore the impact of clustering on network lifetime and find the  cluster size to maximize network lifetime. Present a distributed cluster-head (re-)selection scheme.

\item
Explore the feasibility of clustering in different regions of the cell. 

\item
Propose a load-adaptive multiple access scheme, called $n$-phase CSMA/CA, which provides a tunable tradeoff between energy efficiency and delay by choosing $n$ properly.

\end{itemize}
 

The remainder of this article is organized as follows: In the next section, the system model is introduced. In section III, the clustering design  is presented. The communications protocol design is presented in section IV. In section V, we present the simulation results. Concluding remarks are presented in  section VI.
 
\section {System Model}\label{sys}
Consider a single cell with one base station at the center and a massive number of static nodes which are randomly distributed according to a spatial Poisson point process of intensity $\sigma$. The average number of machine nodes in the cell is $N_t=\sigma \pi R_c^2$, where $R_c$ is the radius of the cell. The machine nodes are battery driven and long battery lifetimes are crucial for them. 
The remaining energy of the $i$th device at time $t_0$ is denoted by $E_i(t_0)$, the average time between two data transmissions by $T_i$, and the average packet size by $D_i$. 
The power consumption of node $i$ in the sleeping  and transmitting modes can be written  as  $P_s$ and $P_{t_i}+P_c$  respectively, where $P_c$ is the circuit power consumed by electronic circuits in the transmission mode and $P_{t_i}$ is the transmit power for reliable data transmission. As illustrated in  Fig. \ref{cycle}, a typical machine node may have different energy consumption levels in different activity modes: data gathering, synchronization, transmission, and sleeping.   The expected lifetime for node $i$ at time $t_0$ is  the average length of one duty cycle times the ratio between the remaining energy at time $t_0$ and the average energy consumption per duty cycle: 
\begin{align}\label{eq.L}
&L_i(t_0)= \frac{E_i(t_0)T_i}{E_s+P_s(T_i-\frac{D_i}{R_i}-T_a)+\frac{D_i}{R_i}(P_c+\xi P_{t_i})},
\end{align}
where $R_i$ is the average expected transmission rate for node $i$, $\xi$ is the inverse of power amplifier efficiency, and $E_s$ is the average energy consumption in each duty cycle for data gathering, 
synchronization, resource reservation, and etc. $T_a$ is the active mode duration for data processing other than transmission as represented in Fig. \ref{cycle}.
Let $\tilde{P}_{t_i}(R_i)=\xi P_{t_i}+\frac{R_i}{D_i}(E_s+P_s(T_i-T_a))$  and $\tilde{P}_c=P_c-P_s$, where ${\tilde{P}}_{t_i}(R_i)$ is strictly convex in $R_i$ if ${P}_{t_i}(R_i)$ is strictly convex. Now, one can rewrite \eqref{eq.L} as
\begin{align}
L_i(t_0)&=\frac{E_i(t_0)T_i}{D_i}\frac{R_i}{{\tilde{P}}_{t_i}+{\tilde{P}}_c} =\frac{E_i(t_0)T_i}{D_i}U_i(R_i),
\end{align}
where the energy efficiency  $U_i(R_i)$ is a strictly quasiconcave function of $R_i$ and one can find the optimal $R_i$ to maximize $U_i(R_i)$ \cite{gm}. Then, the lifetime is proportional to   ${U}_i(R_i)$ and the lifetime maximization  is equivalent to maximizing energy efficiency.  For a given system model where $E_i, T_i, T_a, D_i$, $P_c$, $P_s$, and $P_{t_i}$ are known, the  control parameter  is the average data rate in the uplink transmission $R_i$. The choice of multiple access scheme,  level of contention among nodes for channel access, and the amount of available resources for uplink transmission are the main parameters that determine the average expected data rate of a user, and hence, its expected battery lifetime. One must note that given the set of allocated resources to a node, the link-level energy efficiency can be maximized using the techniques in \cite{gm}, which are not the focus of this paper. In the following, we focus on  network-level energy efficiency. To this end, we will answer the following questions:

\begin{itemize}
\item
How should clusters be formed? 
\item
Which communications protocols should be used for intra-cluster communications, i.e. the communications inside the clusters, and
inter-cluster communications, i.e. the communications between the CHs and the BS? 
\end{itemize}
 
 \begin{figure}[!t]
\centering
\includegraphics[width=3.5in]{ 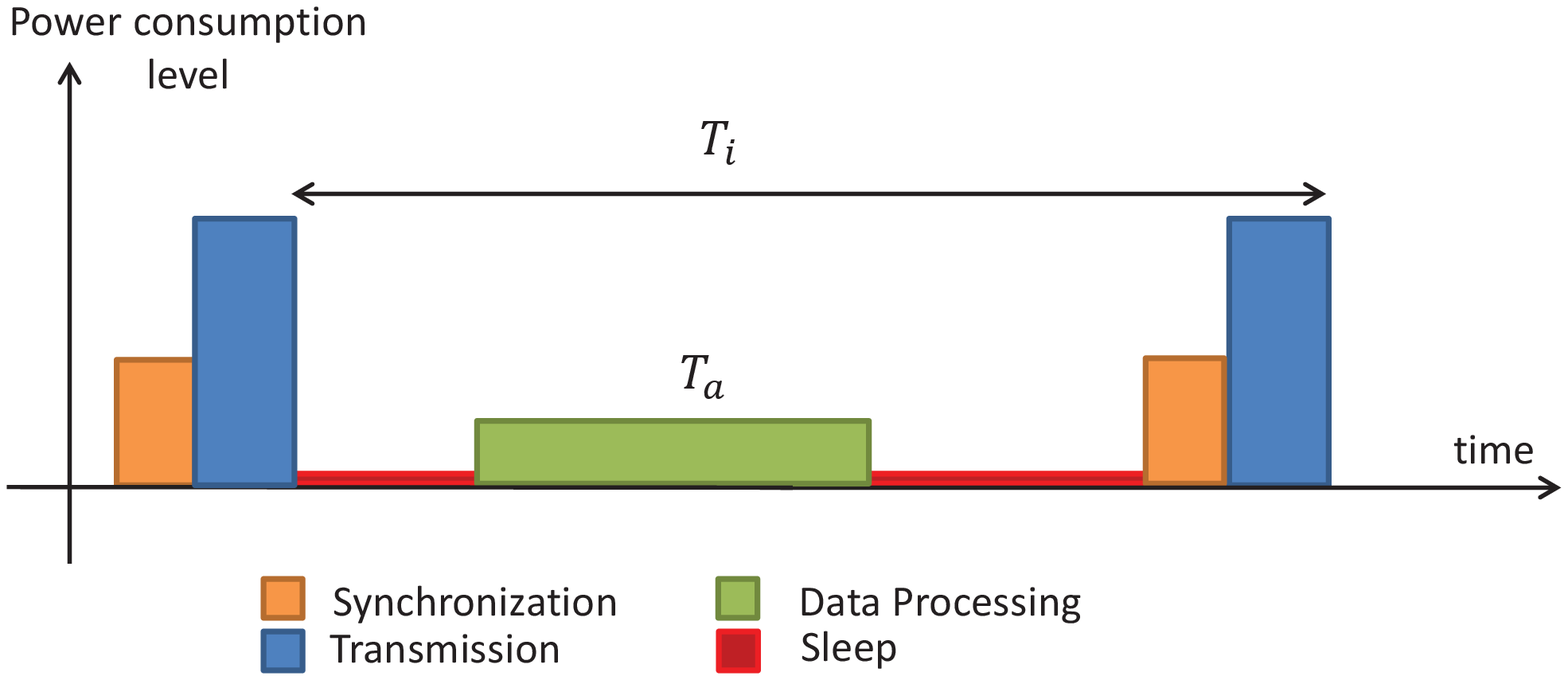}
\caption{Power consumption profile for node $i$. Different modes consume different power levels. } \label{cycle}
\end{figure}

Network lifetime can be defined as a function of individual lifetimes of all machine nodes. Here, we use the {\it {first energy drain}} (FED) network lifetime which is defined as the time at which the first node drains out of energy, and is applicable when missing even one node deteriorates the performance or coverage of the network.  The FED network lifetime is written as $$L_{net}=\min_{i} \hspace{2mm}L_i,$$ where $L_i$ is the lifetime of the $i$th device.    A network that is designed to maximize the FED network lifetime will also minimize the maintenance effort as the interval between  battery replacements in the network is also maximized if a battery is always replaced  once it is dead.

\section{How should clusters be formed?}
With clustering, the number of concurrent channel access requests can be reduced  and the lifetime of cluster members (CMs) can be extended because of less collisions and less transmission power.  However, the lifetime of a cluster head will decrease due to the energy consumption in listening to the channel and relaying packets from its CMs to the BS. Therefore, it is necessary to develop a clustering scheme to improve the overall network lifetime by considering the energy consumptions in both CM and CH nodes. 
 
The clustering problem consists of  finding the number of clusters, and   the CH in each cluster. 
Solving the joint problem is extremely  complicated, if not impossible.  Then, we follow a decoupled approach, define two subproblems, and solve the subproblems sequentially. To this end, in the next subsection we find the number of clusters that should exist in a cell. In subsection \ref{3Pa} we study the problem of finding the CH and the duration of being in the CH mode. 

\subsection{Cluster size}\label{opsi}
Let $p$ denote the probability of being a cluster head for each device, there will be on average $N_tp$ cluster heads in the cell. 
Here, we try to find the probability of being a CH $p$, and hence, the corresponding average cluster-size $z=1/p$, which maximizes the FED network lifetime. 
 To keep the analysis tractable and obtain closed-form expressions, we consider a  homogeneous M2M network in which machine nodes have similar packet lengths and packet generation frequencies.  
 Also, we consider the cluster-forming problem at the reference time where $E_j(t_0)=E_0,\hspace{2mm} \forall j$.  Then, to achieve the highest FED lifetime in each cluster, machine nodes must change their turns in order to avoid that a single node has its energy drained. 
In each duty cycle of the cluster, a node may be in the CH mode with probability $\frac{1}{z}$ and in the CM mode with probability $1-\frac{1}{z}$. Then, the expected lifetime of each node in a cluster which is located at distance $d_h$ from the BS can be expressed as the length of the cluster duty cycle times the ratio between the remaining energy and the average energy consumption in each duty cycle, as follows:
\begin{equation}\label{lifi}
L_c(d_h,z)=\frac{E_0 }{\frac{1}{z} {\mathcal E}_h+(1-\frac{1}{z}){\mathcal E}_m}T_c,
\end{equation}
where the energy consumptions of each node in the CM and CH modes are written as: 
\begin{align}
&  {\mathcal E}_m={E_s+\tilde{D}\frac{P_c+\xi P_t^m}{R_m}},\quad \label{lmm} {\mathcal E}_h={E_s^h+ \frac{(z-1)\tilde{D}}{{R}_m}P_l+ [1+\lambda (z-1)] \tilde{D}\frac{P_c+\xi P_t^h}{R_h}},
\end{align}
respectively. In this expression,
 $\lambda$ is the packet-length compression coefficient at the CH and captures the packet compression effect at the CH which may decode and re-encode the packets of its members for more efficient data transmission. $\tilde{D}$ is the average packet size, $T_c$ is the cluster duty cycle, $P_l$ is the listening power consumption, $\frac{ (z-1)\tilde{D}}{R_m}P_l$ models the energy consumption in receiving packets from the CMs, and $E_s^h$ is the  average static energy consumption in the CH mode which is usually greater than $E_s$  due to the processing and compressing operations  on the received packets from the CMs. Assume the expected data rate function is $F_{\mathcal X}(w, P,\Omega (x),u)$, where  $w$ is the available bandwidth, $P$ the transmit power, $\Omega(x)$ the path loss as a function of distance $x$, $u$ the number of nodes which share the medium, and $\mathcal X$ the multiple access scheme. For example,  if frequency division multiple access (FDMA) and time division multiple access (TDMA) schemes are used, we have \cite{elza}:
\begin{align}
F_{FDMA}(w, P,\Omega (x),u)&=\frac{w}{u}\log(1+\frac{P}{N_0\Gamma \Omega(x)\frac{w}{u}}),\label{FF}\\
\text{and}\hspace{2mm} F_{TDMA}(w, P,\Omega (x),u)&=\frac{w}{u}\log(1+\frac{P}{N_0\Gamma \Omega(x)w}),\label{FT}
\end{align}
respectively, where $N_0$ is the noise power spectral density, and the additional loss term $\Gamma$ is introduced to account  for other losses associated with the specific scenario and the signal to noise ratio (SNR) gap between channel capacity and a practical coding and modulation scheme. Obviously, \eqref{FF} and \eqref{FT} are strictly convex and decreasing in $\Omega(x)$ and $u$, and strictly concave and increasing in $P$ and $w$. In the following, we assume $F_{\mathcal X}(w, P,\Omega (x),u)$ is strictly convex and decreasing in $\Omega(x)$ and $u$, and strictly concave and increasing in $P$ and $w$.  The expected data rates of the CHs and  CMs are found as:\begin{align}  
R_h =F_{\mathcal H}(w_h, P_t^h,\Omega_h (d_h),\frac{N_t}{z}),\quad 
R_m=F_{\mathcal  M}(w_m,P_t^m,\Omega_m ({{d}}_m),z),\nonumber
\end{align}
where $\mathcal  H$ and $\mathcal  M$ are the medium access schemes from the CH to the BS and from the CM to the CH respectively, $w_m$ and $w_h$ are the bandwidths for intra- and inter-cluster communications respectively, and $z$ and $\frac{N_t}{z}$ are the number of nodes which share the intra- and inter-cluster communications' resources respectively.  The inter- and intra-cluster communications path loss functions are modeled as  $\Omega_h (d_h)={\beta_h }({ d}_h)^{\gamma_h}$ and $\Omega_m (d_m)={\beta_m }({ d}_m)^{\gamma_m}$ where ${ d}_m$ is the average distance between CMs and the respective CHs, $d_h$ the average distance between CHs and the BS, $\beta_h$  and $\beta_m$ are constants, and $\gamma_h$ and $\gamma_m$ are path loss exponents.  

Recall that machine nodes are randomly distributed according to a spatial Poisson point process of intensity $\sigma$ in the cell. As each node independently decides to be a cluster-head with probability $p$, one can assume that CHs and CMs are distributed as  independent homogeneous spatial Poisson processes $P_1$ and $P_0$ with intensity parameter $\sigma_1=p \sigma$ and $\sigma_0=(1-p) \sigma$ \cite{14Foss}.
Each non-CH device joins the cluster of its closest CH, then a Voronoi tessellation is formed  in the cell \cite{14Foss} and the cell area is divided into zones called Voronoi cells where each Voronoi cell has a nucleus, i.e. a $P_1$ process which shows the CH. The average number of CMs in each cluster, $\tilde M$, represents the average number of $P_0$ process points in each Voronoi cell and the total length of all segments which connect the $P_0$ process points to the nucleus in a Voronoi cell is denoted by $\tilde J$. Based on the derivations in \cite{sima}, the $\tilde M$ and $\tilde J$ are derived as  $\tilde M= \frac{1-p}{p}$, and $ \tilde J = \frac{1-p}{2p^{\frac{3}{2}}\sqrt{\sigma}}$, respectively. Now, the average distance between a cluster member and its respective cluster head is derived as
\begin{equation} \label{eq:EX3}
{ d}_{m}=  {\tilde J}/{\tilde M}={1}/({2 \sqrt{\sigma p}}) = \sqrt{\frac{z}{4\sigma}}.
\end{equation}
Now, one can rewrite  the lifetime expression in \eqref{lifi} for $\lambda=1$ as  follows:
\begin{align}
&L_c (d_h,z)=\label{FEDl} \frac{E_0 T_c}{E_s\text{+}\frac{E_s^h-E_s}{z}\text{+} \frac{(z-1)\tilde D (P_l+P_c+\xi P_t^m)}{z F_{\mathcal M}(w_m,P_t^m,\Omega_m (\sqrt{\frac{z}{4\sigma}}),z)}+\frac{\tilde D (P_c+\xi P_t^h)}{F_{\mathcal H}(w_h, P_t^h,\Omega_h (d_h),\frac{N_t}{z})}}.
\end{align}
 Then,  the  cluster-size that maximizes \eqref{FEDl} is found as:
\begin{equation}\label{opz}
z^*=\frac{1}{p^*}=\arg \max_z\hspace{1.5mm} \min_{d_h}\hspace{1.5mm} L_c(d_h,z), 
\end{equation}
which maximizes the minimum cluster lifetime in the network. As the minimum cluster-lifetime happens in the cell edge, i.e. $d_h= R_c$, the optimization problem  in \eqref{opz} reduces to:
\begin{equation}\label{opz1}
z^*=\frac{1}{p^*}=\arg \max_z\hspace{1.5mm} L_c(R_c,z).
\end{equation}
For example, when $\mathcal X=\mathcal Y= \text{FDMA}$,  \eqref{FEDl} reduces to:
\begin{align}
L_c& (d_h,z)=\label{FEDl0}\frac{E_0 T_c}{E_s+\frac{E_s^h-E_s}{z}+\frac{\tilde D (z-1)(P_c+\xi P_t^m+P_l)}{w_m \log(1+A_1 z^{(1-\frac{\gamma_m}{2})})}+\frac{N_t \tilde D (P_c+\xi P_t^h)}{z w_h\log (1+A_2/ z)}}, 
\end{align}
 in which $A_1=\frac{P_t^m(4\sigma)^{\frac{\gamma_m}{2}}}{\Gamma N_0 w_m \beta_m}$ and $A_2=\frac{P_t^hN_t }{\Gamma N_0 w_h \beta_h( d_h)^{\gamma_h}}$.
One sees maximizing $L_c(d_h,z)$ in \eqref{FEDl0} is equivalent to minimizing  its denominator. Also by taking the second derivative of  the denominator of $L_c$ in \eqref{FEDl0} with respect to $z$, one can see that it is a strictly convex function over $z>0$ and $2\le \gamma_m\le 4$, which are typical for intra-cluster communications. Then, using the convex optimization tools, the proposed cluster size in \eqref{opz1} can be found. The $z^*$ in \eqref{opz1} is the desired cluster size at the reference time when all CMs inside a cluster have the same remaining energy levels, i.e. $E_0$ in \eqref{FEDl}. In subsection \ref{3Pa}, we will present a CH reselection scheme that balances the energy consumptions of all CMs so that  their remaining energy levels are as close to each other as possible. Then, we can use \eqref{FEDl} to estimate the desired cluster size at any time instant by replacing $E_0$ with the respective remaining energy level.

\subsection{Cluster-head (re)selection for FED maximization} \label{3Pa} 
After deriving the probability of being a CH, the BS broadcasts $p^*$ to all machine nodes in the cell. Then, $N_tp^*$ of them  broadcast themselves as the initial CHs and the remaining nodes  are connected to the nearest CH. In order to maximize the FED lifetime in each cluster, the existing CH in each cluster can gather position information and communication characteristics of its respective CMs and finds a new CH for its respective cluster. This information can be sent in regular intervals or on demand along with the ordinary data from the CMs to their respective CHs.  Equivalently, the existing set of CHs can send the gathered information to the BS and let the BS to derive the new set of CHs.

Define the set of machine nodes which are grouped in a given cluster as $\Psi $, and  the duty cycle of the cluster  as $T_c$. Recall the lifetime expression for the $i$th machine node at time $t_0$ from \eqref{eq.L}. Our aim here is to select a CH at time $t_0$ to maximize the minimum individual lifetime of the clustered nodes. Define the index of the selected CH as $i^*(t_0)$. The selected node must satisfy the following condition:
\begin{align}
&L_{net}(\text{using}\hspace{2mm} i^*)\ge L_{net}(\text{using any}\hspace{1.5mm} j\in \Psi)
\longrightarrow \min_{i\in \Psi} \frac{E_i(t_0)T_c}{\mathcal E_{i,i^*}}\ge \min_{i,j\in \Psi} \frac{E_i(t_0)T_c}{\mathcal E_{i,j}},\label{eqop}
\end{align}
where $\mathcal E_{i,k}$ is the expected energy consumption of node $i$ in each duty cycle of operation, defined as follows:
\begin{align}
\mathcal E_{i,k}=\left\{ \begin{array}{l}
  {E_s+D_i {(P_c+\xi P_t^m)}/{R_m^{i,k}}} \hspace{24 mm}\text{if $i$   $\ne k,$}\\
 {E_s^h+ \frac{\psi\tilde{D}}{ R_m^{\upsilon ,i}}P_l+ [1+\lambda \psi] \tilde{D}\frac{P_c+\xi P_t^h}{R_h^{i,b}}}\hspace{13mm}\text{if $i$ $=  k$,}
\end{array} \right.\nonumber
\end{align}
$k$ is the respective CH of node $i$, $\psi=|\Psi|-1$  is the number of CMs in $\Psi$, $R_m^{i,k}$ the average data rate between node $i$ and node $k$, $R_h^{i,b}$ the average data rate between node $i$ and the BS, and $R_m^{\upsilon,k}$ the average intra-cluster communications data rate. The data rate functions are found as:
\begin{align}
R_m^{i,k}=&F_{\mathcal M}(w_m,P_t^m,\Omega_m(d_{i,k}),z),
R_h^{i,b}=F_{\mathcal H}(w_h,P_t^h,\Omega_h(d_{i,b}),\frac{N_t}{z}),\nonumber\\
R_m^{\upsilon,k}=&F_{\mathcal M}(w_m,P_t^m,\Omega_m(d_{\upsilon,k}),z).\nonumber
\end{align}
In these expressions, $d_{i,k}$ is the distance between node $i$ and node $k$, $d_{i,b}$ is the distance between node $i$ and the BS, and $d_{\upsilon,k}$ is the average distance from an arbitrary point in the cluster to node $k$. Based on the cluster shape, one can use the average distance results in \cite{sadr} to find the appropriate estimate of $d_{\upsilon,k}$. For example, if the cluster shape can be approximated by a circle with radius $R$, the average distance to  node $k$ which is located at distance $r$ from the cluster center  is given by 
\begin{equation}\label{eqas1}
d_{\upsilon,k}\simeq\frac{2}{3}R+\frac{r^2}{2R}-\frac{r^4}{32R^3}.\end{equation}
Now, we need to estimate $R$ for a given density of nodes and cluster size.   
Define $R_{seg}$ as a random variable to represent the length of the segment from a randomly selected point inside a circle to the center of the circle, where the circle is located at $(0,0)$, and has a radius of $R_{circ}$. The expected value of $R_{seg}$ is derived as:
\begin{equation}\label{eq:j}
{\bar R}_{seg} = \int_{x}\int_{y} \sqrt{x^2+y^2} \frac{1}{\pi R_{circ}^2} \text{d}x\text{d}y = \frac{2}{3} R_{circ},
\end{equation}
where $(x,y)$ shows the position of the selected point with regard to the origin. Recall from \eqref{eq:EX3}, where we have derived the average distance between a CM and its initial CH, which is located at the cluster center as $d_m=\sqrt{{z}/{4\sigma}},$ in which $z$ and $\sigma$ show the cluster size and density of nodes, respectively. Then, if one estimates the shape of constructed clusters inside a cell with circle, the average radius of the constructed clusters can be estimated by combining \eqref{eq:EX3} and \eqref{eq:j},  as follows: 
\begin{equation}\label{eq_rad}
R=\frac{3}{2}d_m=\frac{3}{2}\sqrt{\frac{z}{4\sigma}}.
\end{equation}
The derived $R$ in \eqref{eq_rad} can be employed subsequently in \eqref{eqas1} in order to derive an approximation of $d_{\upsilon,k}$. 
%
%
%
%
%
%
%
In light of the above derivations, one can find the index of the desired CH as:
\begin{equation}\label{mm}
i^*(t_0)=\arg \max_{i\in \Psi} \big( \min_{j\in \Psi}\hspace{1.5mm} \frac{E_j(t_0)T_c}{\mathcal E_{j,i}} \big).
\end{equation} 
From  \eqref{eqop} and  \eqref{mm}, one sees that the choice of the CH is dependent upon: (i) the remaining energy of devices, and hence, it is time-dependent; (ii) the distance between machine devices; (iii) the distance between each device and  the BS; and (iv) the average length of the queued data at each device. If adjacent triggers for CH reselection are too closely placed, then it may result in energy wasting as no change in the CH selection is needed in multiple consecutive periods. If adjacent triggers are too far apart, then  negative impact on the network lifetime is possible as a  previously selected CH might be non-optimal in some periods.

\begin{proposition}\label{p2}
The expected CH duration for CH $i^*(t_0)$  is $KT_c$, where $K$ is the smallest non-negative integer that satisfies the following condition for any $j\in\Psi$:
\begin{align}
&\frac{E_{m(i^*)}(t_0) -K \mathcal E_{m(i^*),i^*} }{\mathcal E_{m(i^*),i^*} }  < \frac{E_{m(j)}(t_0) -K \mathcal E_{m(j),i^*} }{\mathcal E_{m(j),j} } ,\label{cone}
\end{align}
and $m(i)$ is the index of node with the shortest expected lifetime when $i$ is the CH, as follows:\begin{equation}\label{mi}m(i)=\arg\min_{j\in \Psi}  \frac{E_j(t_0)T_c}{\mathcal E_{j,i}}.\end{equation}

\end{proposition}
\begin{proof}
As $i^*$ is the selected CH at $t_0$, it satisfies the necessary condition in \eqref{eqop} which can be rewritten  as follows:
\begin{align}
\min_{i\in \Psi}  \frac{E_i(t_0)T_c}{\mathcal E_{i,i^*}}\ge &  \min_{i, j \in \Psi}  \frac{E_i(t_0)T_c}{\mathcal E_{i,j}}\longrightarrow
\frac{E_{m(i^*)}(t_0)  }{\mathcal E_{m(i^*),i^*} }  \ge  \frac{E_{m(j)}(t_0) }{\mathcal E_{m(j),j} } \label{eqre}.
\end{align}
\eqref{eqre} shows that Proposition \ref{p2} is true for $K=0$. If  $i^*$ is the respective CH of node  $i$ in time interval $[t_0,t_0+\kappa T_c]$, the expected remaining energy of node $i$ at time $t_0+\kappa T_c$ is $E_i(t_0) -\kappa \mathcal E_{i,i^*}.$ Then,  $i^*$ is the desired CH at $t_0+(K-1)T_c$ since
\begin{align}&\frac{E_{m(i^*)}(t_0) -\kappa \mathcal E_{m(i^*),i^*} }{\mathcal E_{m(i^*),i^*} }\ge \frac{E_{m(j)}(t_0) -\kappa \mathcal E_{m(j),i^*} }{\mathcal E_{m(j),j} }, \hspace{10mm} \forall j\in \Psi, \forall \kappa\in\{0,\cdots,K-1\}.\label{cone0}
\end{align}
At time $KT_c$, there exists a $j\in \Psi$ such that:
\begin{align}
&\frac{E_{m(i^*)}(t_0) -K \mathcal E_{m(i^*),i^*} }{\mathcal E_{m(i^*),i^*} }< \frac{E_{m(j)}(t_0) -K \mathcal E_{m(j),i^*} }{\mathcal E_{m(j),j} }.\nonumber
\end{align}
Then, node $j$ will  be the desired CH beyond $t_0+KT_c$, and hence, we have Proposition \ref{p2}.
\end{proof}
 In practice, frequent CH reselections may introduce high signaling overhead. Less frequent CH reselections can be used instead with some performance losses. Fig. \ref{fedli} presents the cumulative distribution function (CDF) of  individual lifetimes of a group of 10 clustered machine nodes   for different CH reselection periods. One sees that by applying the proposed CH selection scheme in \eqref{mm} {\it fast enough} so that the CH will be reselected whenever the result of \eqref{mm} is changed, the minimum individual lifetime is maximized and all nodes will die almost at the same time. Then, their batteries can be replaced at the same time, thus minimizing the human interventions and the efforts of maintaining the network.

 \begin{definition}\label{de121}
A feasible selection of the CH is {\it max-min fair } if an increase in the individual lifetime of any node must be at the cost of a decrease of some
already smaller lifetime \cite[chapter~4]{gzor}.
\end{definition}

\begin{proposition}
By applying the proposed CH selection scheme in \eqref{mm} fast enough, the max min fairness of the lifetimes of all CMs can be maintained. \end{proposition}
\begin{proof}
From \eqref{mm}, one sees that the selection of CH $i^*$ achieves the max-min individual lifetime.  Denote  node with the shortest expected lifetime when $i^*$ is the CH as the bottleneck node, where its index can be found from \eqref{mi} as $m(i^*)$. Then, if we select any node other than $i^*$ as the CH to increase the lifetime of a given node, the expected lifetime of of the bottleneck decreases, and hence, the selected CH in \eqref{mm} satisfies the max-min fairness requirement in Definition \ref{de121} for a limited CH duration as discussed in Proposition \ref{p2}. Then, if we reselect the CH fast enough, i.e. whenever the result of \eqref{mm} changes, the max min fairness of the lifetimes of all CMs can always be maintained. 
\end{proof}
 
By maintaining the max min fairness of the CMs' lifetimes, machine nodes will either have the same lifetime  or die earlier because of limited energy storage at the beginning.  The latter case happens when a machine node has a very low  initial remaining energy level and it dies earlier than the others even if never serves as the CH\footnote{The interested reader may refer to section 4.2.4 in \cite{gzor} for more information.}.  Quantitative analysis for the former case, where all CMs have the same initial remaining energy levels, is presented in Fig. \ref{cfed}. One sees that by successive CH reselections the minimum expected lifetime is increased and the maximum expected lifetime is decreased, and hence, the  difference which is depicted by a red-colored curve converges to zero.

\begin{figure}
        \begin{subfigure}[b]{0.5\textwidth}
                \includegraphics[width=3.3in]{ 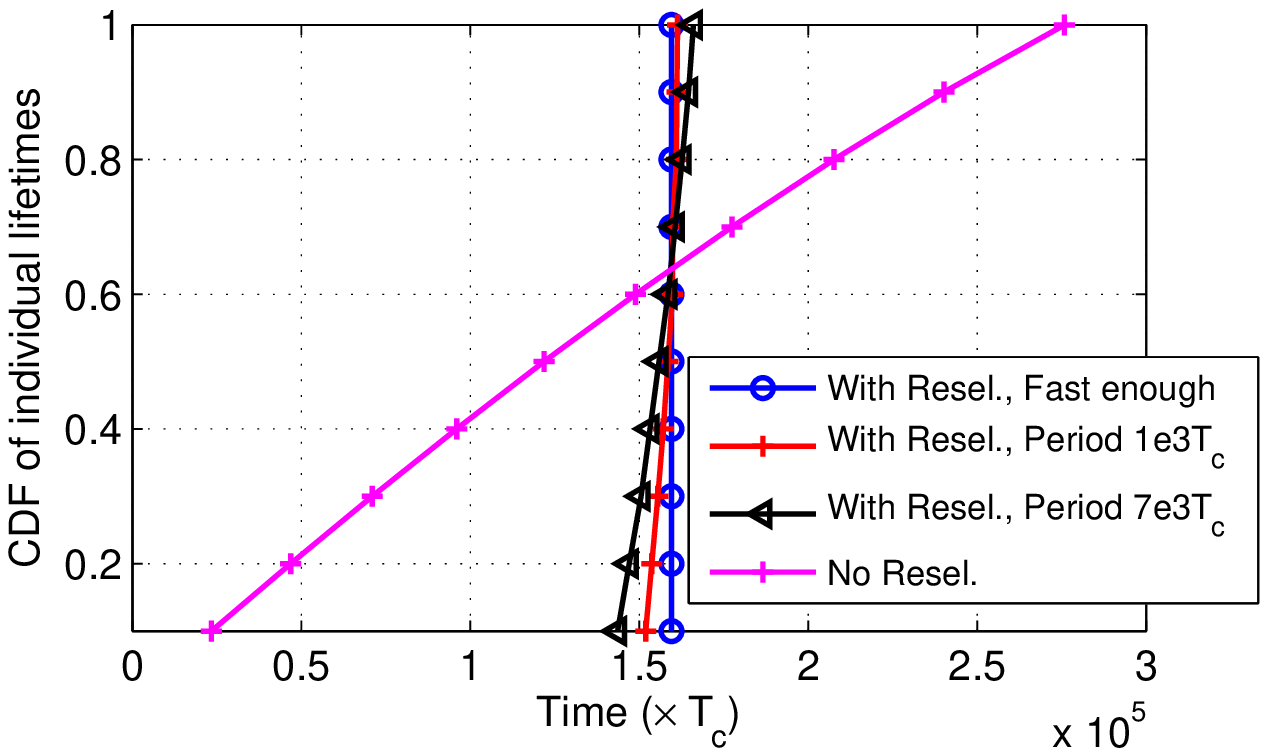}
                \caption{CDF of individual lifetime of machine nodes for different  CH reselection periods.}
                \label{fedli}
        \end{subfigure}
         ~
         \begin{subfigure}[b]{0.5\textwidth}
                \includegraphics[width=3.3in]{ 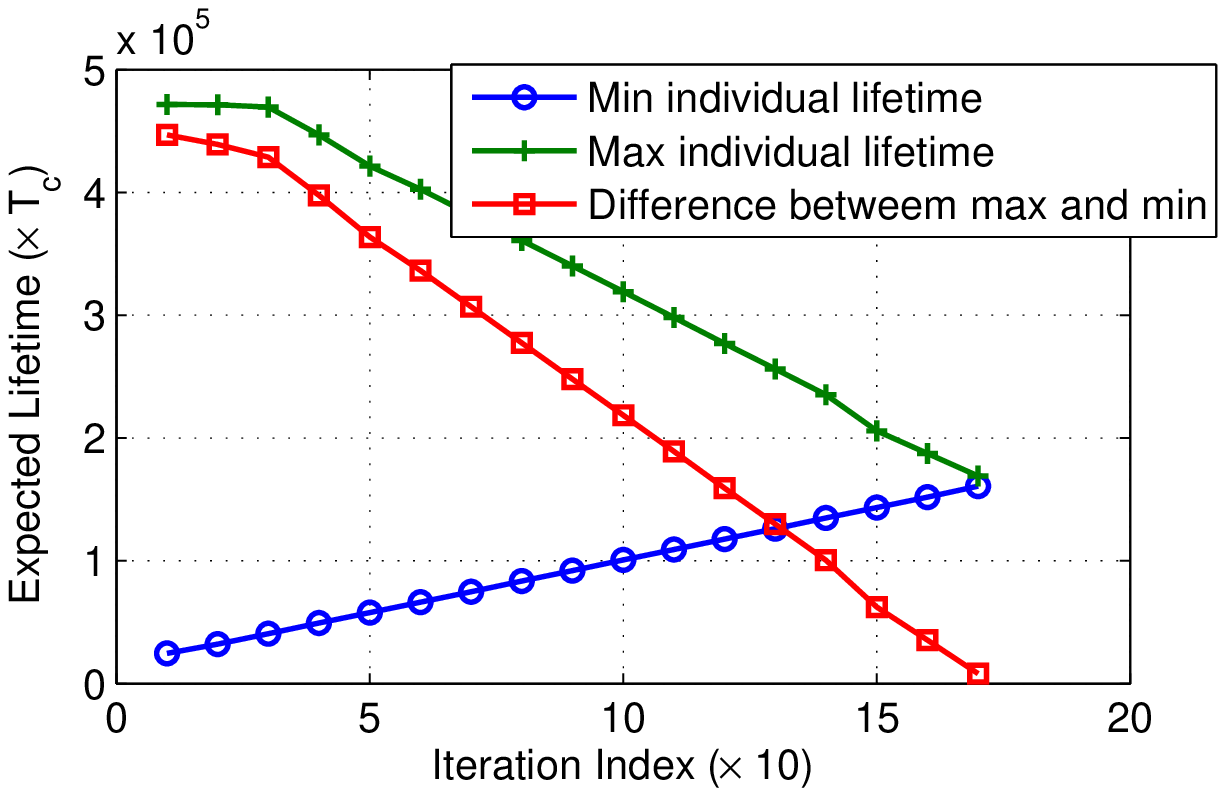}
                \caption{Maximum and minimum of individual lifetimes versus iteration index. Reselection period $ = 1000 T_c.$}
                \label{cfed}
        \end{subfigure} \\
             \caption{Performance evaluation of the proposed CH selection scheme for 10 clustered nodes.  Cluster radius = 50 m, distance from cluster center to the BS\hspace{0.5 mm}=\hspace{0.5 mm}250 m,  $D_i=1$ KByte  $\forall i\in \Psi$, $\xi=2$,  $P_t^h=0.2$ W, $P_t^m=0.05$ W, $P_c=.02$ W, $\Gamma$=13 dB, $w_h\hspace{0.5 mm}=\hspace{0.5 mm}0.4 w_m$\hspace{0.5 mm}=\hspace{0.5 mm}360 KHz.}\label{fig_fed}
\end{figure}

 \subsection{Cluster Reformation}\label{refor}
Here, we investigate the impact of reforming clusters on the network lifetime. As mentioned in the previous subsection, the initial CHs are located at the cluster centers. By reselecting the CHs, a newly selected CH can be located at the cluster border, and hence, the average communications distance to this CH will be higher than the case in which CH is located at the cluster center. In this case, reforming the clusters may improve the energy efficiency of intra-cluster communications if the energy cost for reforming the clusters is low. When a CH is located at the cluster center, the average communications distance to the CH is derived from \eqref{eq:EX3} as $d_{cent}=\sqrt{{z}/{4\sigma}}.$ However, if a node which is located at distance $r$ from the cluster center is selected as the CH, the average communications distance to the CH is derived from \eqref{eqas1} as 
$d_{r}\simeq 0.5/x+{2}r^2x/3-0.25 r^4x^3$, and $x=\sqrt{{ \sigma}/{z}}.$ 
One sees that in the latter case, the average communications distance has been increased approximately by
$2r^2 x/3,$ 
and hence, the average energy consumption increases accordingly.
Denote the Euclidean distance between a CH and its respective cluster center by $r$, the CH duration by $T_{dur}$, the average duty cycle of the connected devices by $T_{c}$, and the average  energy  cost per device for reforming the clusters  by $E_{ref}$. Then, reforming the clusters will save energy if:
$$E_{ref}< \frac{T_{dur}}{T_{c}}[\mathcal E_m^{\text{if ref.}} -\mathcal E_m^{\text{if not ref.}}].$$
In this expression,  $\mathcal E_m^{\text{if ref.}}$ and $\mathcal E_m^{\text{if not ref.}}$ are derived from \eqref{lmm} as:
$$\mathcal E_m^{\text{if ref.}}={E_s+\tilde{D}\frac{P_c+\xi P_t^m}{F_{\mathcal  M}(w_m,P_t^m,\Omega_m (d_{cent}),z)}},\quad \mathcal E_m^{\text{if not ref.}}={E_s+\tilde{D}\frac{P_c+\xi P_t^m}{F_{\mathcal  M}(w_m,P_t^m,\Omega_m (d_r),z)}},$$
respectively. Then, in the case that CH re-selection is performed in long intervals, i.e. $T_{dur}$ is large in comparison with $T_{c}$, and the selected CH is far from the cluster center,  joint CH re-selection and reforming the clusters can further prolong the network lifetime. In section \ref{sim}, we evaluate the impact of  $E_{ref}$ on the feasibility of cluster reforming.

\subsection{Where  should clustering be used?}\label{whref}
 In section \ref{opsi} we have investigated the cluster-size problem for machine nodes uniformly distributed in a cell. In practice, the density of nodes may vary from one place to another.  Then, in order to deploy an M2M solution in a specified region, e.g. smart metering in a building, it is crucial to  investigate the impact of clustering on the network lifetime. \\
 Consider the system model in Fig. \ref{fwh} where the region of interest is shown in gray and $N$ machine devices are planned to be deployed in this region. The radius of this region and the average distance from this region to the BS are denoted by $r$ and $R$ respectively.
Clustering should be used in this region when the FED network lifetime can be improved. Using derivations in section \ref{opsi}, the expected FED lifetime of an M2M network with and without clustering is found as: 
\begin{align}
L_c=\frac{E_0T_c}{\frac{1}{N} {\mathcal E}_h^c +\frac{N-1}{N}{\mathcal E}_m^c}, \quad L_d=\frac{E_0T_c}{{\mathcal E}_h^d},
\end{align}
where
\begin{align}
&{\mathcal E}_h^c=E_s^h+\frac{(1+\lambda(N-1)) \tilde D(P_c+\xi P_t^h)}{F_{\mathcal H}(\bar w_h,P_t^h,\Omega_h(R),1)}+\frac{\tilde DP_l(N-1)}{F_{\mathcal M}(w_m,P_t^m,\Omega_m(\bar r),N-1)},  \nonumber\\
&{\mathcal E}_m^c=E_s+\frac{\tilde D(P_c+\xi P_t^m)}{F_{\mathcal M}(w_m,P_t^m,\Omega_m(\bar r),N-1)},  {\mathcal E}_h ^d=E_s^d+\frac{\tilde D(P_c+\xi P_t^d)}{F_{\mathcal H}(w_m+\bar w_h,P_t^h,\Omega_h(R),N)}.  \nonumber
\end{align}
In these expressions, ${\mathcal E}_h^d$ is the average energy consumption in the direct access to the BS and is assumed to be the same for all nodes in the region of interest, ${ E}_s^d$ and $P_t^d$ are the static energy consumption and the transmit power in direct access mode, $w_m$ and $\bar w_h$ are the allocated bandwidths to the CMs and CH respectively, and ${\mathcal E}_h^c$ and ${\mathcal E}_m^c$ are the average energy consumptions in CH and CM modes respectively. Then, to check the feasibility of clustering  one need to check if $L_d<L_c$ is satisfied. Let us derive a tractable necessary condition for the feasibility of clustering in a special case, where $P_l=P_c$, $\mathcal X=\mathcal Y=\text{FDMA}$, and the transmit powers are set to achieve the  predefined average SNRs $s_h$ and $s_b$ at the CH and  BS respectively. Clustering is used when:
\begin{align}
&\hspace{0mm} L_c>L_d,\nonumber\\
&\rightarrow \mathcal E_0-P_cQ\tilde D+\frac{P_t^d N^2\tilde D \xi}{w_t\log(1+s_b)})>\quad\frac{M \tilde D\xi P_t^h}{w_h\log(1+s_b)}+\frac{(N-1)^2\tilde D\xi P_t^m}{w_m\log(1+s_h)}\nonumber,\\
&\rightarrow     {\bar s_b}   \Omega_h(R)(N-M)> {\bar s_h} (N-1) \Omega_m(\bar r)+\frac {P_cQ\tilde D-\mathcal E_0}{\Gamma N_0 \tilde D\xi},\label{cong1} 
\end{align}
where $\bar r$ is the average distance between two random points in a circle with radius $r$, and is found as $\frac{128 r}{45 \pi}$ \cite{avp}. Also,
\begin{align}
 \mathcal  E_0=&NE_s^d-E_s^h-(N-1)E_s; \quad
Q=\frac{M}{w_h\log(1+ s_b)}+\frac{2(N-1)^2}{w_m\log(1+s_h)}-\frac{N^2}{w_t\log(1+ s_b)},\nonumber\\
M=&1+\lambda(N-1); w_t=w_h+w_m; \bar s_x=\frac{s_x}{\log(1+s_x)}, \quad x\in\{b,m\}.\nonumber
\end{align}
Solving the inequality in \eqref{cong1} for $M\ne N$, we have:
\begin{align}
 \Omega_h(R)>  \frac{{\bar s_h} (N-1)}{{\bar s_b}(N-M)} \Omega_m(\bar r) +\frac {P_cQ\tilde D-\mathcal E_0}{{\bar s_b}\Gamma N_0 \tilde D\xi(N-M)}.\label{cong}
\end{align}
The inequality derived in \eqref{cong} represents the general condition which must be satisfied in any region where clustering is feasible. 

%

\begin{figure}[!t]
\begin{subfigure}[t]{0.45\textwidth}
\centering
\includegraphics[width=3 in]{ 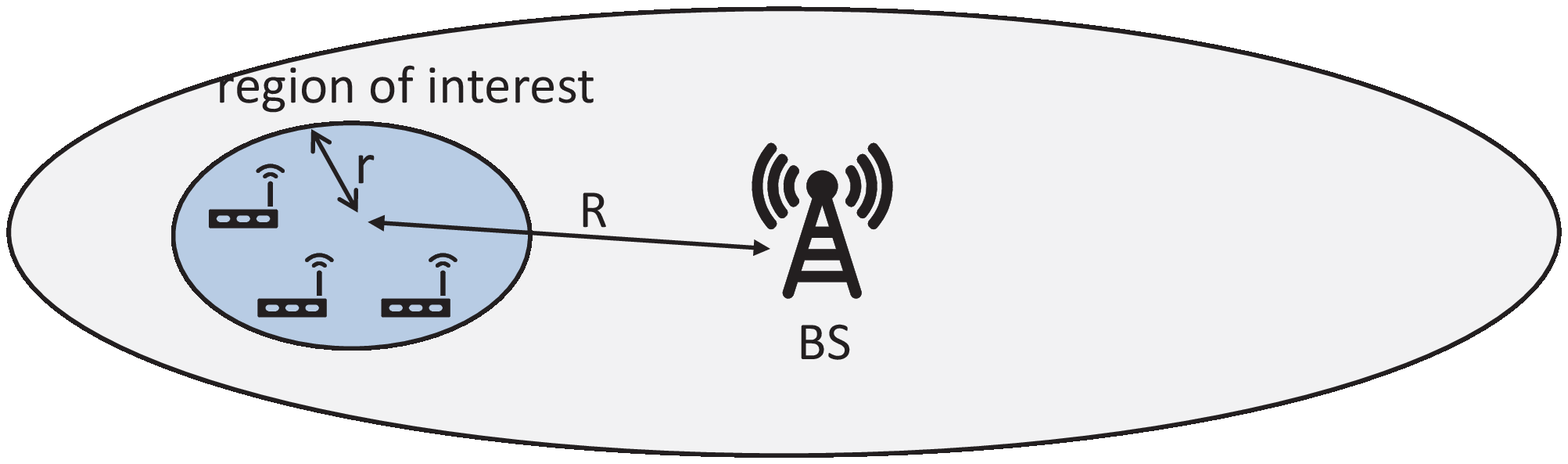}
\caption{Region of interest in the cell } \label{fwh}
\end{subfigure}
~
 \begin{subfigure}[t]{0.55\textwidth}
                \includegraphics[width=3.5in]{ 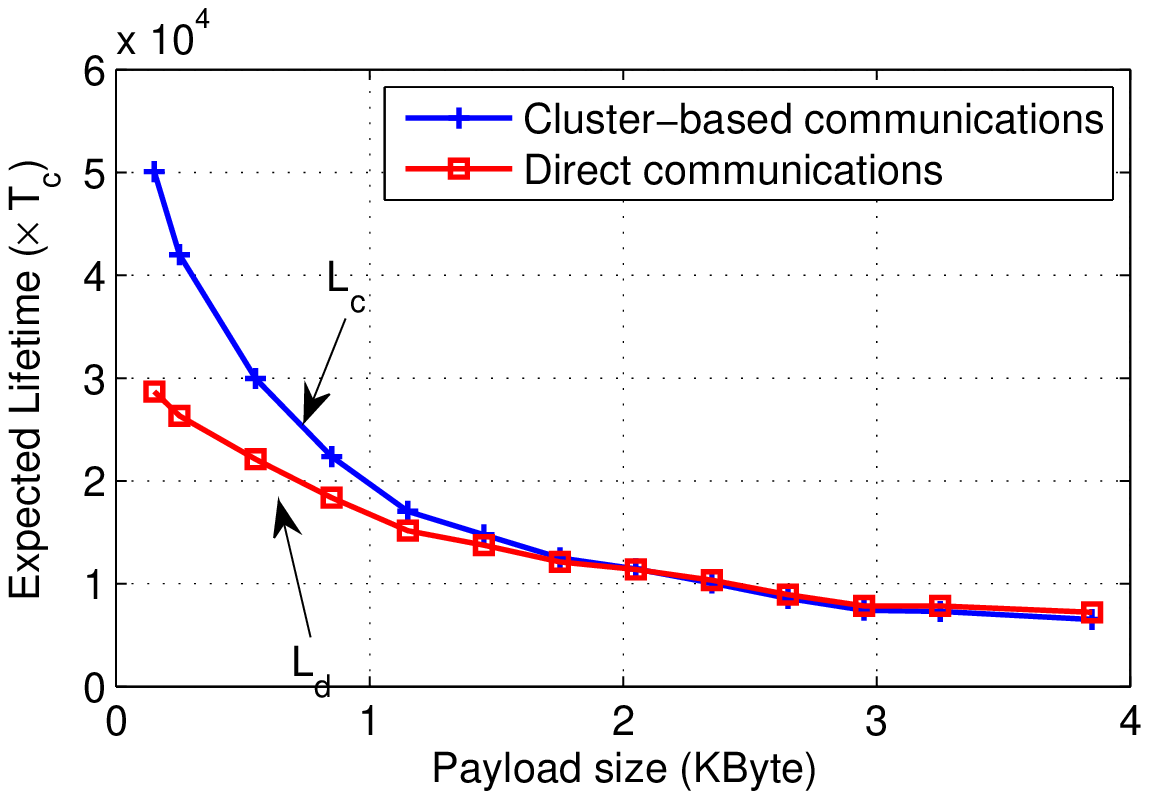}
        \caption{ Network lifetime versus payload size.  $s_h\hspace{0.5 mm}=\hspace{0.5 mm} s_b=\hspace{0.5 mm}20$ dB, 
 and $\mathcal E_0$= 16 mJ. Other simulation parameters are the same as Fig. \ref{fig_fed}.}\label{ax1}
\end{subfigure}
\caption{Investigation on the feasibility of clustering in different regions of the cell.}\label{wcf}
\end{figure}

From \eqref{cong}, one can conclude that the increase in the  cluster size, circuit power consumption, and required SNR at the CH may result in the infeasibility of clustered communications. For any setup that $L_c< L_d$,  clustering can not prolong the network lifetime. One may decrease the number of clustered nodes by making multiple clusters in order to make the clustered communications feasible. In a multi-cell scenario, out-of-cell interference is also a limiting factor which may affect the feasibility of clustering in cell-edge regions where adjacent clusters reuse the same set of time/frequency resources. In this case, machine nodes that observe high interference power may communicate directly with the BS.  Fig. \ref{ax1} presents the FED network lifetime for a group of 10 clustered machine nodes versus payload size, when $\lambda=1$, i.e. the CH does not compress the CM's packets. In this figure, one sees when the payload size goes beyond 2.1 KBs, the direct communications approach outperforms the cluster-based communications approach.  In order to evaluate the tightness of the above proposed necessary conditions for clustering, we predict the crossover point of Fig. \ref{ax1} by solving   \eqref{cong1} for $\tilde D$,
\begin{align}
&\frac {\mathcal E_0-P_cQ\tilde D}{\Gamma N_0 \tilde D\xi}  +{\bar s_b}   \Omega_h(R)(N-M)> {\bar s_h} (N-1) \Omega_m(\bar r)\longrightarrow  \tilde D<16584 \text{ bits}=2.02\hspace{0.5mm} \text{KB},  \label{neq}
\end{align}
where the pathloss functions are given in Table \ref{sim3}.
 The predicted crossover point in  \eqref{neq} matches well with the simulation results in Fig. \ref{ax1}.

\begin{figure*}[!t]
\centering
\includegraphics[width=6 in]{ 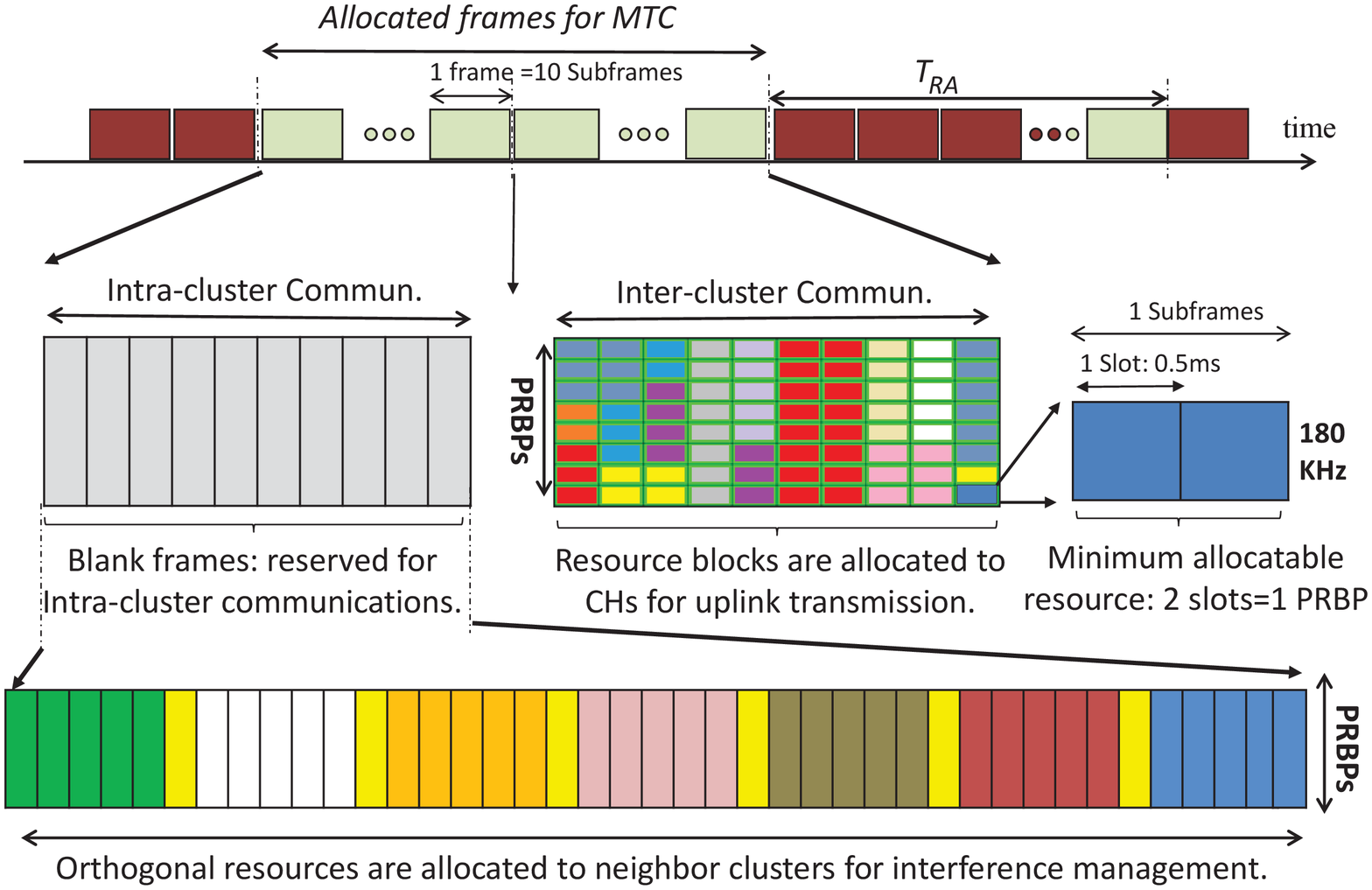}
\caption{ The proposed $E^2$-MAC for LTE systems. } \label{Lframe}
\end{figure*}



\section{Energy efficient medium access}\label{asr}
In this section we investigate an energy efficient medium access protocol for M2M communications. 
The communications consist of two phases: (i) intra-cluster communications from CMs to CHs and (ii) inter-cluster communications from CHs and non-clustered nodes to the BS. The two phases may use orthogonal resources e.g. different time slots or different frequency bands. Fig. \ref{Lframe} illustrates a potential frame structure for LTE systems when the two phases use different time resources. In the first phase, all cluster members  send data to their cluster heads. Then, the CHs will forward the data to the BS in the second phase. Also, intra-cluster communications can be an underlay to inter-cluster communication, i.e. uplink resources can be reused for intra-cluster communications, and this is out of the scope of this paper, and the interested reader may refer to \cite{intsh} for details.  
 
Inter-cluster communications from the CHs to the BS may happen either in asynchronous or synchronous mode and should follow  existing cellular standards. In the case of LTE \cite{sched},  the typical way for  asynchronous connection to the BS is the RACH, as discussed in section \ref{rsec}.  In the synchronous mode,  connected devices send their scheduling requests to the BS through the physical uplink control channel (PUCCH). The BS performs the scheduling and sends back
the scheduling grants through the corresponding physical downlink
control channel (PDCCH) for each node. Now, the granted machine nodes are able to send data over the granted Physical Uplink Shared Channel (PUSCH).  
Energy efficient scheduling can be implemented at the BS to further improve the lifetime of the CHs.  The interested reader may refer to our previous works in \cite{vtc}-\cite{wcnc} for more information. 

 In the following, we focus on intra-cluster communications. If the number of clusters in a cell is limited, BS may allocate orthogonal time/frequency resources to the clusters for intra-cluster communications.   In a realistic massive MTC deployment, it could occur the case where there is not enough orthogonal resources and therefore the clusters may reuse the same resources for intra-cluster communications.  The interference from adjacent clusters in the same or nearby cells can be dealt with using link level or network level techniques. For example, a machine node can increase its transmission power when it observes high interference power or use lower modulation order so that it's more robust to interference, and vice versa. From the network level perspective, most interference management schemes which have been standardized for heterogeneous cellular networks with several femotocells  deployed in a macro cell, e.g. almost blank subframe (ABS) \cite{3gppn}, and frequency planning can be used for interference avoidance between clusters. Besides, random access based approaches can be used for intra-cluster communications to further avoid interference between adjacent clusters.  The proposed $E^2$-MAC in Fig. \ref{Lframe} benefits from an interference-aware resource allocation scheme for intra-cluster communications. Depending on the cluster-size, and hence the traffic load in each cluster, the available resources for intra-cluster communications are divided into several bunches of orthogonal resources. Then, these orthogonal resources are allocated to neighbor clusters in order to reduce the received interference at the CHs. Also, the BSs can exchange interference-coordination information with neighbor cells in order to mitigate the inter-cell interference for cell-edge clusters.
  
  Inside each cluster, since only a portion of  machine nodes might be active in each time interval, the communications protocol for intra-cluster communications needs to to be scalable and able to adapt to the changes in the communications needs of the active nodes.    Among the proposed protocols in literature, CSMA/CA is a promising approach for intra-cluster communications as it does not need additional control overhead and can adapt to the changes in the number of connected nodes \cite{bo2}. In addition, CSMA/CA has the potential of avoiding interference from neighbor clusters.    In the sequel, we investigate the  energy efficiency of CSMA/CA and its shortcomings in high traffic-load regimes. To overcome the shortcomings and further improve the energy efficiency of the network, we introduce the $n$-phase CSMA/CA.


\subsubsection{Energy efficiency of non-persistent CSMA/CA}\label{bse}
Different transmission techniques can be used in CSMA/CA, for example 1-persistent CSMA/CA, p-persistent CSMA/CA, non-persistent CSMA/CA, or
the RTS/CTS mechanism. Here, we focus on  non-persistent CSMA/CA because of its low cost in implementation. Non-persistent CSMA/CA has been standardized in IEEE 802.15.4 for low data rate solutions like ZigBee and WirelessHART \cite{anana}. In non-persistent CSMA/CA, a machine node waits for a random amount of time after sensing a busy channel and repeats this algorithm until finding the channel idle to transmit data. In the following, we analyze the energy efficiency of non-persistent CSMA/CA.

\begin{figure}[!t]
\begin{subfigure}[t]{0.45\textwidth}
\centering
\includegraphics[width=1.5in]{ 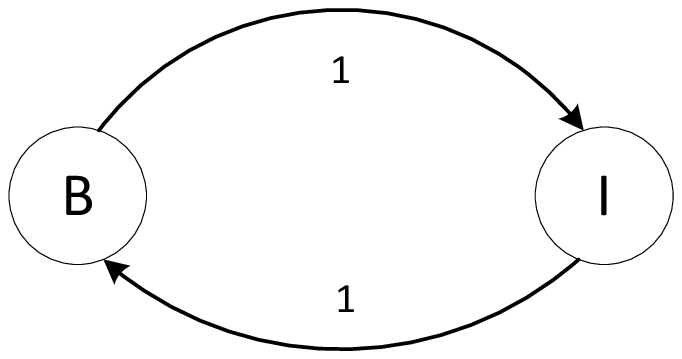}
\caption{State transitions in non-persistent CSMA/CA.}\label{mcnp}
\end{subfigure}
~
 \begin{subfigure}[t]{0.55\textwidth}
\includegraphics[width=3.5in]{ 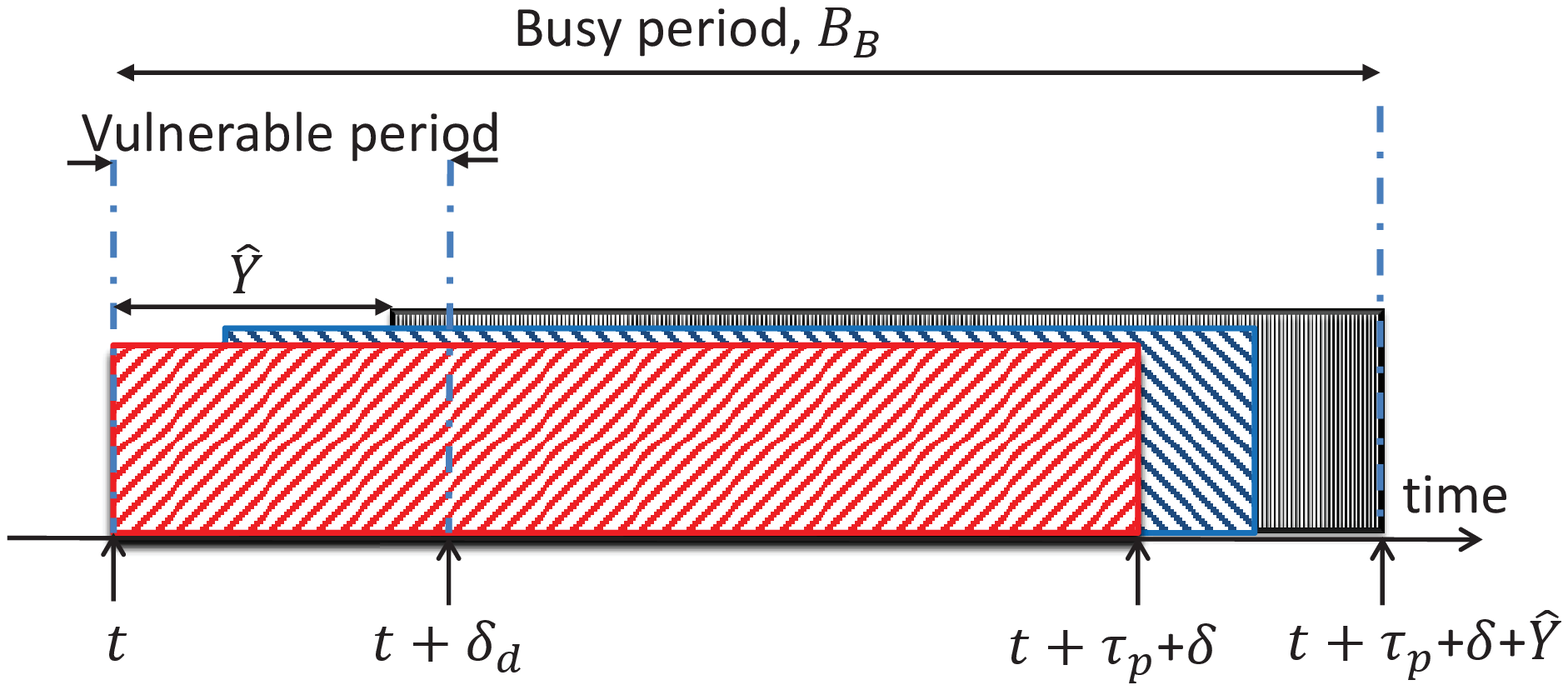}
\caption{Busy transmission period in non-persistent CSMA/CA \cite{map}.}\label{mark}
\end{subfigure}
\caption{Idle and busy periods of non-persistent CSMA/CA}\label{ooomark}
\end{figure}

Define the aggregated packet arrival rate of all machine nodes in a cluster as $g$, which includes both new arrivals and retransmitted ones. We assume that the acknowledgment packets are transmitted in an independent channel to simplify the analysis.  
There are two states of channel utilization: idle and busy. In the busy state, the transmission can be either successful of unsuccessful. The channel utilization is modeled as a two-state Markov process as shown in Fig. \ref{mcnp}. The probability of each possible transition between states is 1.  Based on this model, the probabilities of the idle and busy states are the same, i.e. $ \pi_I=\pi_B=0.5$. The average duration of the idle state is the average time between two  consecutive packets, i.e. $B_I=1/g$. Define $\tau_p$ and $\delta_d$ as the transmission and detection delay. The average duration of the busy period is $B_B=\tau_p+\delta+\hat{Y}$ where $\delta$ is the propagation delay. Also, $\hat{Y}$ denotes the average time at which the last interfering packet is scheduled within a transmission period that started at time 0, as illustrated in Fig. \ref{mark}.  $\hat{Y}$ is calculated as follows:
\begin{align}
F_{{Y}}(y)&=pr(\text{no arrival during $\delta_d-y$})= e^{-g(\delta_d-y)},\nonumber\\
\text{and} \quad \hat{Y}&=\delta_d-({1-e^{-g\delta_d}})/g.
\end{align}
Packet transmission will be successful if it starts after an idle period and no other node starts transmission after that. The time-averaged idle channel probability, which represents the probability that the channel is  idle when a new packet arrives in the network, is derived as: 
\begin{align}
p_i&=\frac{\pi_IB_I}{\pi_IB_I+\pi_BB_B}=\frac{1/g}{1/g+T-(1-e^{-g\delta_d})/g}=1/({gT+e^{-g\delta_d}}),
\end{align}
where $T=\tau_p+\delta_d+\delta$.
Also, the probability of no-transmission after the transmission of a tagged packet is the probability of no-transmission in $\delta_d$, and is derived as $p_s=e^{-g\delta_d}$. Then, the probability of successful packet transmission being happening when a new packet arrives in the network is the multiplication of time-averaged idle channel probability, $p_i$, and no collision after that, $p_s$, as follows:
\begin{align}
p_{is}&=p_i\times p_s=\frac{1/g \times e^{-g\delta_d}}{1/g+ \tau_p+\delta+(\delta_d-(1-e^{-g\delta_d})/g)}=1/({g T e^{g\delta_d}+1}).\label{psnp}
\end{align}
The average amount of consumed energy for each new packet that arrives  in the network is calculated as:
\begin{align}
{ E}_{cons}&=(1-p_i){ E}_{B}+p_i(1-p_s) { E}_F+p_i p_s { E}_S,
\end{align}
 where ${ E}_S$ models the energy consumption in a successful packet transmission, ${ E}_F$ models the energy consumption  in an  unsuccessful packet transmission, and ${ E}_B$ models the energy consumption after a busy sensed channel, as follows:
\begin{align}
{ E}_S&=(P_c+\xi P_{t}^m)\tau_p+P_l\tau_r, 
{ E}_F={ E}_S+P_l\theta_f, \text{ and}\quad
{ E}_B=P_l\theta_b.\label{be}
\end{align}
In \eqref{be}, $\theta_b$ and $\theta_f$ are the average backoff after sensing a busy channel and collision respectively,  and $\tau_r$ is the round-trip-time delay from successful packet transmission to the acknowledgment packet arrival.
 Then, one can derive the energy efficiency of the network for intra-cluster communications as follows:
\begin{align}
U_E(g)&=\frac{\tilde D\hspace{0.5mm} p_{is}}{ { E}_{cons}}=\frac{\frac{\tilde D}{({gTe^{g\delta_d}+1})}}{\frac{{ E}_S}{1+gTe^{g\delta_d}} +\frac{ gTe^{2g\delta_d}}{(1+gTe^{g\delta_d})^2}{ E}_F+(1-\frac{e^{g\delta_d}}{1+gTe^{g\delta_d}}){ E}_B}\nonumber\\
&=\frac{\tilde D }{{ E}_S+\frac{gTe^{2g\delta_d}}{gTe^{g\delta_d}+1}{ E}_F+(1+(gT-1)e^{g\delta_d}){ E}_B}.\label{ue}
\end{align} 
 The throughput of the network for intra-cluster communications is derived by finding the portion of time in which successful transmission happens, as follows:
\begin{align}
U_S(g)&=\frac{\pi_B\tau_p p_s}{\pi_BB_B+\pi_IB_I}R_{in}
 =\frac{ge^{-g\delta_d}\tau_p}{gT+e^{-g\delta_d}}R_{in}=\frac{g\tau_p}{1+gTe^{g\delta_d}}R_{in},\label{us}
\end{align}
in which
\begin{equation} 
R_{in} = w_m\log(1+\frac{P_t^m}{N_0\Gamma \Omega_m(d_m)w_m}).\nonumber
\end{equation}
One can see that the expression in \eqref{us} quite matches the throughput analysis in section 4.1 of \cite{map}. 
Define packet delay as the time interval between packet arrival and successful transmission. Then, the average packet delay is derived by considering the average time spent in backoffs and retransmissions before a successful packet transmission, as follows:

\begin{align}
{D}_{c}(g)&= \label{d0p}\sum\limits_{k=0}^{k_m }(1-p_{is})^k p_{is}\bigg[\tau_p+k\big(\frac{1-p_i}{1-p_{is}}\theta_b
+p_i\frac{1-p_s}{1-p_{is}}(\theta_f+\tau_p)\big)\bigg],\\
& \mathop \approx\limits^{k_m\gg 1}  \tau_p+(\frac{1}{p_{is}}-1)\bigg[\frac{1-p_i}{1-p_{is}}\theta_b
+p_i\frac{1-p_s}{1-p_{is}}(\theta_f+\tau_p)\bigg],\nonumber
\end{align}
where $\frac{1-p_i}{1-p_{is}}$ and $p_i\frac{1-p_s}{1-p_{is}}$ are the probabilities of unsuccessful transmission due to a busy sensed channel and collision respectively. Also, $k_m$ is the maximum number of times that a machine node tries to transmit a specific packet.
%


The energy efficiency, spectral efficiency, and delay performance  of a CSMA/CA-based system are depicted in Fig. \ref{figt}. In Fig. \ref{figt1}, one sees that the energy efficiency and delay performance of the system degrade in the traffic load. This is due to the fact that the probability of collision increases in the traffic load. Also, one sees that the spectral efficiency of the system increases in the traffic load in low to medium traffic loads, and decreases in the traffic load in high traffic loads. Taking the first derivative of $U_S$ in \eqref{us} with respect to $g\tau_p$, one sees that  the spectral efficiency is maximized when 
\begin{equation}\label{lm}g\tau_p=\frac{2}{a} \text{LambertW}( \sqrt{a}/2).\end{equation}
In this expression, LambertW function is the inverse of the function $f(x) = x\exp(x)$, $a=\frac{\delta_d}{\tau_p}$, and $a\ll 1$ is assumed. Inserting $a=0.005$ in \eqref{lm}, one sees that  $U_S$ is maximized when $g\tau_p=13.7$ which matches well with the simulation results. Fig. \ref{figt2} shows the tradeoffs between energy and spectral efficiency, and delay and spectral efficiency when $g\tau_p\le 10$. One sees that any improvement in the spectral efficiency of the system is achieved at the cost of degradation in the  energy efficiency and  delay performance of the system.

In the following section, we present a load-adaptive hybrid TDMA/CSMA protocol, called $n$-phase CSMA/CA, which offers a tunable trade-off between energy efficiency, spectral efficiency, and  delay performance of the network. 
%

\begin{figure}[t!]
    \begin{subfigure}[t]{0.5\textwidth}
        \includegraphics[trim={0cm 0cm 0.9cm 0cm},clip,width=3.3in]{ 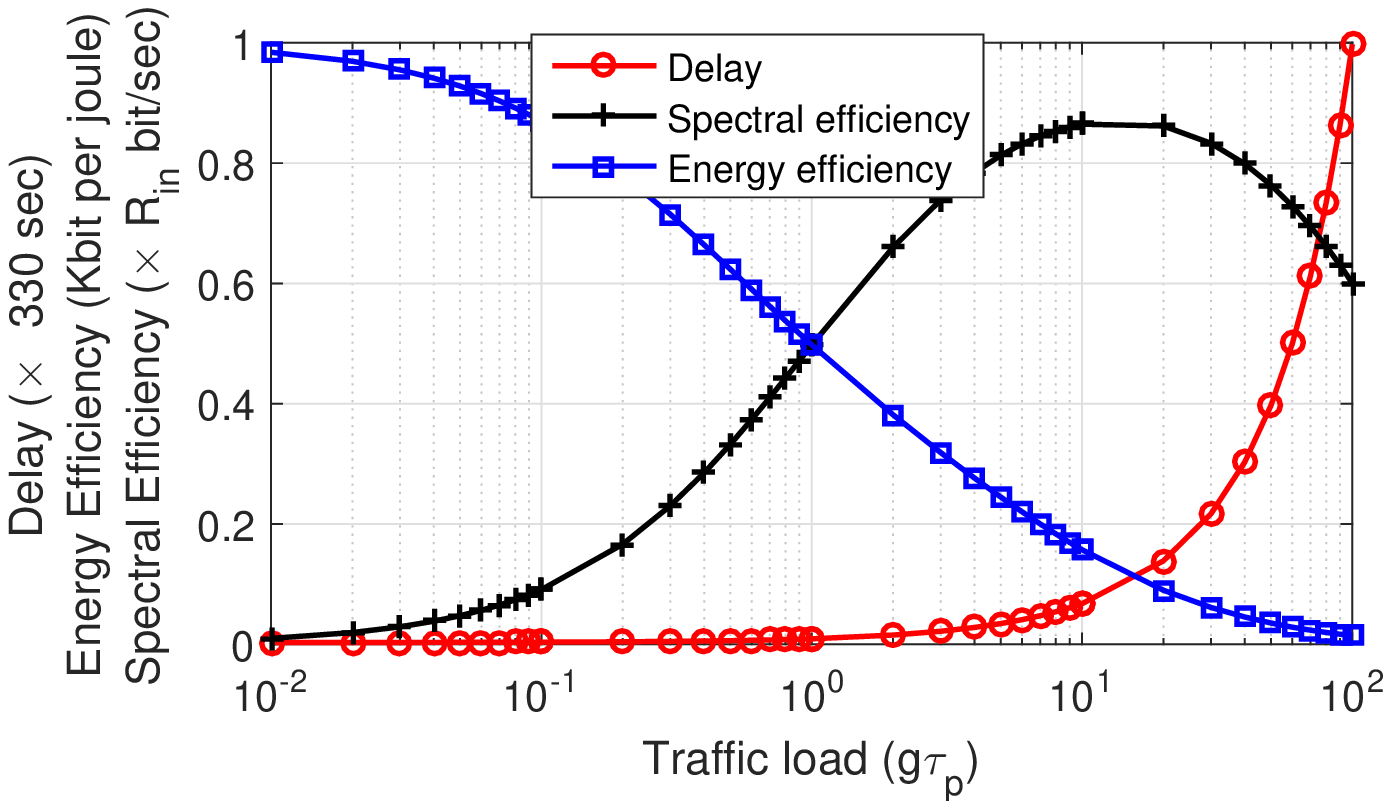}
        \caption{Energy  efficiency, delay, and spectral efficiency versus traffic load}
        \label{figt1}
    \end{subfigure}%
~
    \begin{subfigure}[t]{0.5\textwidth}
        \includegraphics[trim={0cm 0cm 0.9cm 0cm},clip,width=3.3in]{ 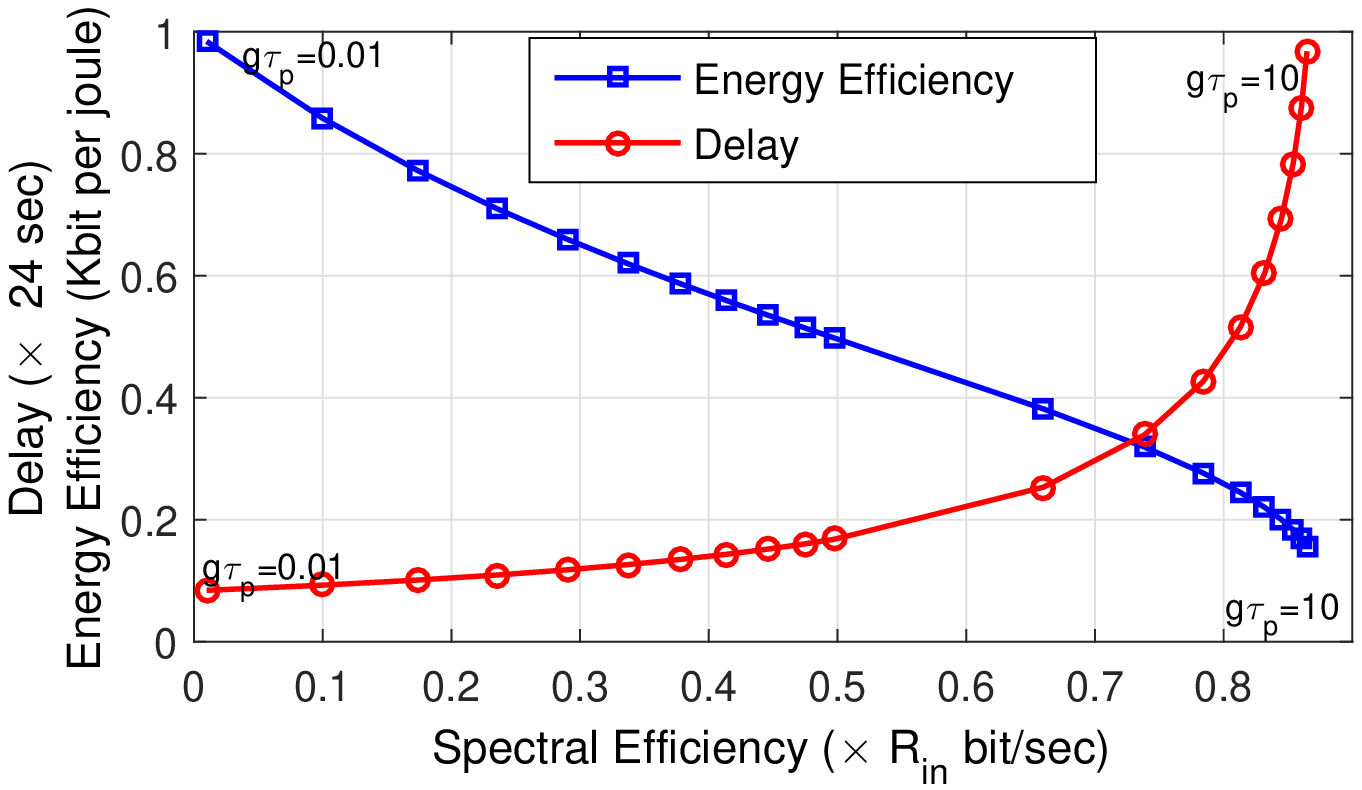}
        \caption{Energy  efficiency and delay versus spectral efficiency}
\label{figt2}
    \end{subfigure}
\caption{Energy  efficiency, delay, and spectral efficiency of a CSMA/CA-based system.  Parameters: $T= 1$ sec, $\tilde D=5$, $\delta_d/\tau_p=0.005$, $E_B=2$ mJ, $E_S=5$ mJ, and $E_F=6$ mJ. } \label{figt}
\end{figure}

\begin{figure*}[t!]
    \centering
    \begin{subfigure}[t]{0.5\textwidth}
        \centering
        \includegraphics[width=2.8in,height =2.8in]{ 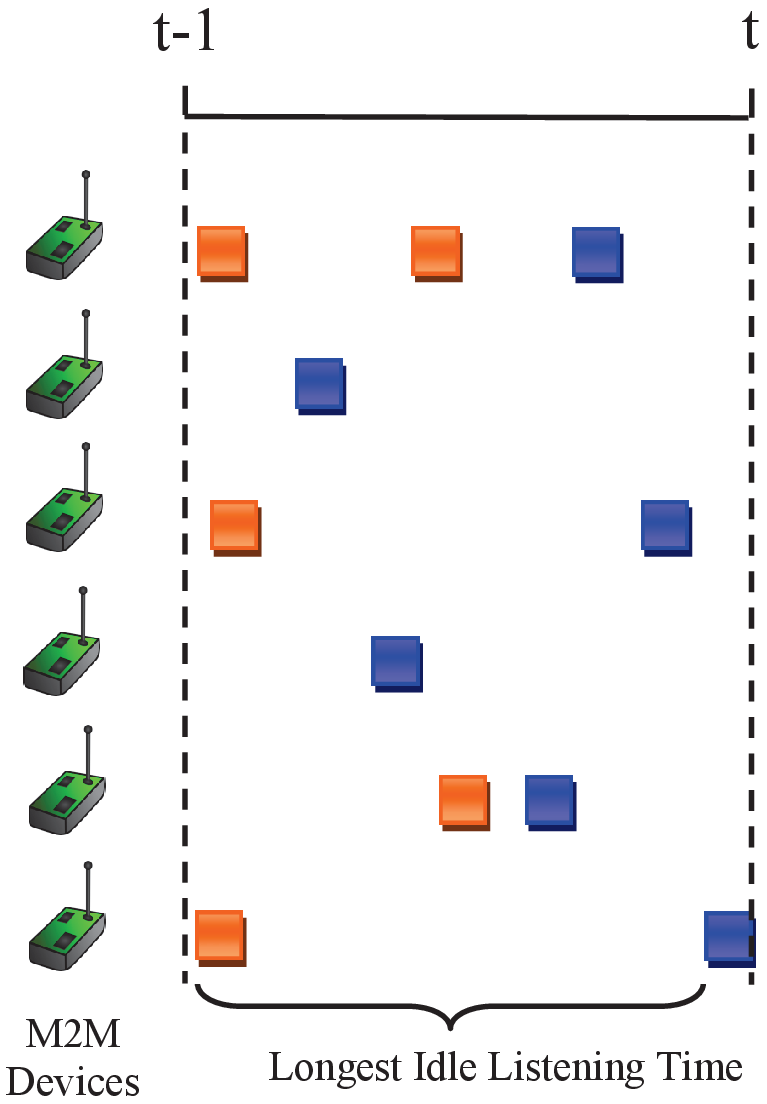}
        \caption{Ordinary CSMA/CA}
    \end{subfigure}%
~
    \begin{subfigure}[t]{0.5\textwidth}
        \centering
        \includegraphics[width=2.8in,height =2.8in]{ 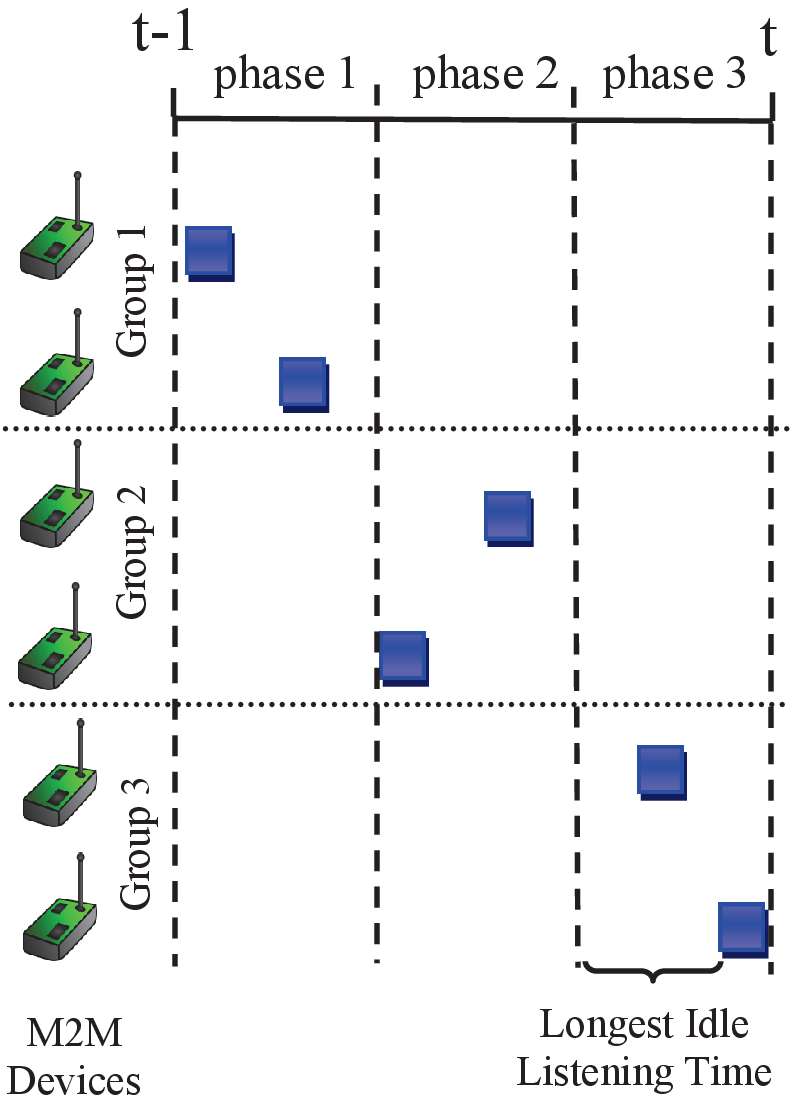}
        \caption{$n$-phase CSMA/CA (n=3)}
    \end{subfigure}
\caption{Ordinary CSMA/CA and $n$-phase CSMA/CA when $n$=3.  Red- and blue-colored squares show failed and successful transmissions respectively. The idle listening time and collisions are decreased in the $n$-phase CSMA/CA scheme significantly. } \label{cs}
\end{figure*}

%

\subsubsection{$N$-phase CSMA/CA}\label{mul}
The major drawback of the non-persistent CSMA/CA is its inherent inefficiency in the high traffic-load regime, i.e. increasing traffic load prolongs the idle-listening time and decreases the successful transmission probability, thus wastes energy. To solve the issue, we try to reduce the contention among nodes. To this end, we present a flexible and load-adaptive multiple  access protocol, called $n$-phase CSMA/CA, which divides each contention interval into $n$ phases, as illustrated in Fig. \ref{cs}. In each phase, only a portion of the CMs are permitted to compete for channel access. Before the assigned phase starts, each node keeps sleeping instead of listening  and newly arrived packets are buffered. 
Note that when $n=1$, it is the same as the conventional CSMA/CA. When $n$ is sufficiently large, at most one user will be assigned to each phase and it is the same as the scheduling-based MAC. Therefore, $n$-phase CSMA/CA provides a tradeoff between 
contention- and scheduling-based medium access schemes. 
By choosing an appropriate $n$, the probability of successful packet transmission can be increased to reduce both the number of collisions and idle listening time to achieve the desired energy efficiency. To explore the impact of $n$ on the performance of the network, in the  following we derive the energy efficiency, delay, and spectral efficiency as a function of $n$. 
 
 By using $n$-phase CSMA/CA, the available users will be divided among $n$ phases, and hence, the corresponding traffic load in each phase  will be $g_n\simeq\frac{g}{n} $.  Then, the energy  and spectral efficiency of the network using $n$-phase CSMA/CA are derived  from \eqref{ue}-\eqref{us} as:
\begin{align}
U_E(g_n)&=\frac{\tilde D }{{ E}_S+\frac{g_nTe^{2g_n\delta_d}}{g_nTe^{g_n\delta_d}+1}{ E}_F+(1+(g_nT-1)e^{g_n\delta_d}){ E}_B}\nonumber\\
\text{and}& \quad U_S(g_n)=\frac{g_nT+e^{-g_n\delta_d}-1}{1+g_nTe^{g_n\delta_d}}R_{in}.\nonumber
\end{align}
Also, the average packet delay for $n$-phase CSMA/CA is derived as follows:
\begin{align}
{D}_{nc}(g_n)&=\label{dnp}\sum\limits_{k=0}^{k_m }(1-{\tilde p}_{is})^k {\tilde p}_{is}\bigg(\tau_p+k\big(\frac{1-{\tilde p}_i}{1-{\tilde p}_{is}}\theta_b
+{\tilde p}_i\frac{1-{\tilde p}_s}{1-{\tilde p}_{is}}(\theta_f+\tau_p)\big)\bigg) \\
& \mathop \approx\limits^{k_m\gg 1}  \tau_p+(\frac{1}{\tilde p_{is}}-1)\big(\frac{1-{\tilde p}_i}{1-{\tilde p}_{is}}\theta_b
+{\tilde p}_i\frac{1-{\tilde p}_s}{1-{\tilde p}_{is}}(\theta_f+\tau_p)\big)\nonumber
\end{align}
where ${\tilde p}_i=\frac{1}{n}\frac{1}{{g_nT+e^{-g_n\delta_d}}}$, ${\tilde p}_s=e^{-g_n\delta_d}$, ${\tilde p}_{is}={\tilde p}_{i}{\tilde p}_{s}$.

\subsubsection{Performance tradeoff of $n$-phase CSMA/CA}\label{ptra}
 Fig. \ref{mph} represents the tradeoff between energy efficiency, spectral efficiency, and delay performance of a network with different numbers of phases. By increasing the number of phases, the probability of successful transmission increases which results in higher energy efficiency due to a less number of retransmissions and shorter time spending  in idle-listening mode. In the same time, one sees that the average packet delay increases in the number of phases because of packet buffering until the assigned slot starts. Furthermore, the spectral efficiency of  network decreases as the number of phases increases. The presented tradeoff in  Fig. \ref{mph} shows how  one can sacrifice the delay  and spectrum efficiency performance of the network to enable energy efficient M2M communications, and hence, achieve higher levels of battery lifetimes.  For example in the case of delay-constrained applications, one can find the appropriate number  of phases by choosing the maximum $n$ which satisfies the delay constraint.

\subsubsection{Performance tradeoff of  $n$-phase CSMA/CA with zero detection delay}
When $\delta_d$ is negligible, the $U_E$, $U_S$, and $D_{nc}$ expressions  can be rewritten as:
\begin{align}
U_E(g_n)& \approx \frac{\tilde D }{{ E}_S+\frac{g_nT}{g_nT+1}{ E}_F+g_nT{ E}_B}, \quad U_S(g_n) \approx \frac{g_nT}{1+g_nT}R_{in},\label{aue}\\
{D}_{nc}(g_n)&\approx \sum\limits_{k=0}^{k_m }(1-\frac{1/n}{1+g_nT})^k \frac{1/n}{(1+g_nT)}\big(\tau_p+k\theta_b\big) \label{aud}\mathop \approx\limits^{k_m\gg 1}  \tau_p+\big( n(1+g_nT)-1\big)\theta_b.
\end{align}
If $U_s^{n}$ denotes  the normalized energy efficiency to $R_{in}$, one can derive the tradeoff between energy and spectral efficiency as:
\begin{align}
U_E& \approx \frac{\tilde D }{{ E}_S+U_S^n { E}_F+\frac{U_S^n}{1-U_S^n}{ E}_B}.\label{tr}
\end{align}
From \eqref{tr}, one sees how increasing spectral efficiency $U_S^n$ results in energy efficiency reduction. Similarly, one can derive the tradeoff between delay and spectral efficiency as:
\begin{align}{D}_{nc}\approx\tau_p+(n-1)\theta_b+n\theta_b\frac{U_S^n}{1-U_S^n}\label{trr},\end{align}
from which, we see that the packet delay increases in the spectral efficiency, and hence, the delay performance of the system degrades as the spectral efficiency improves.

 The novel contention-division concept  in the $n$-phase CSMA/CA can be applied in other contention-based protocols, e.g. ALOHA and 802.11, to improve their energy  efficiency.

\begin{figure}[!t]
  \centering
  \begin{minipage}[b]{0.45\textwidth}
    \includegraphics[width=1.1\textwidth]{ 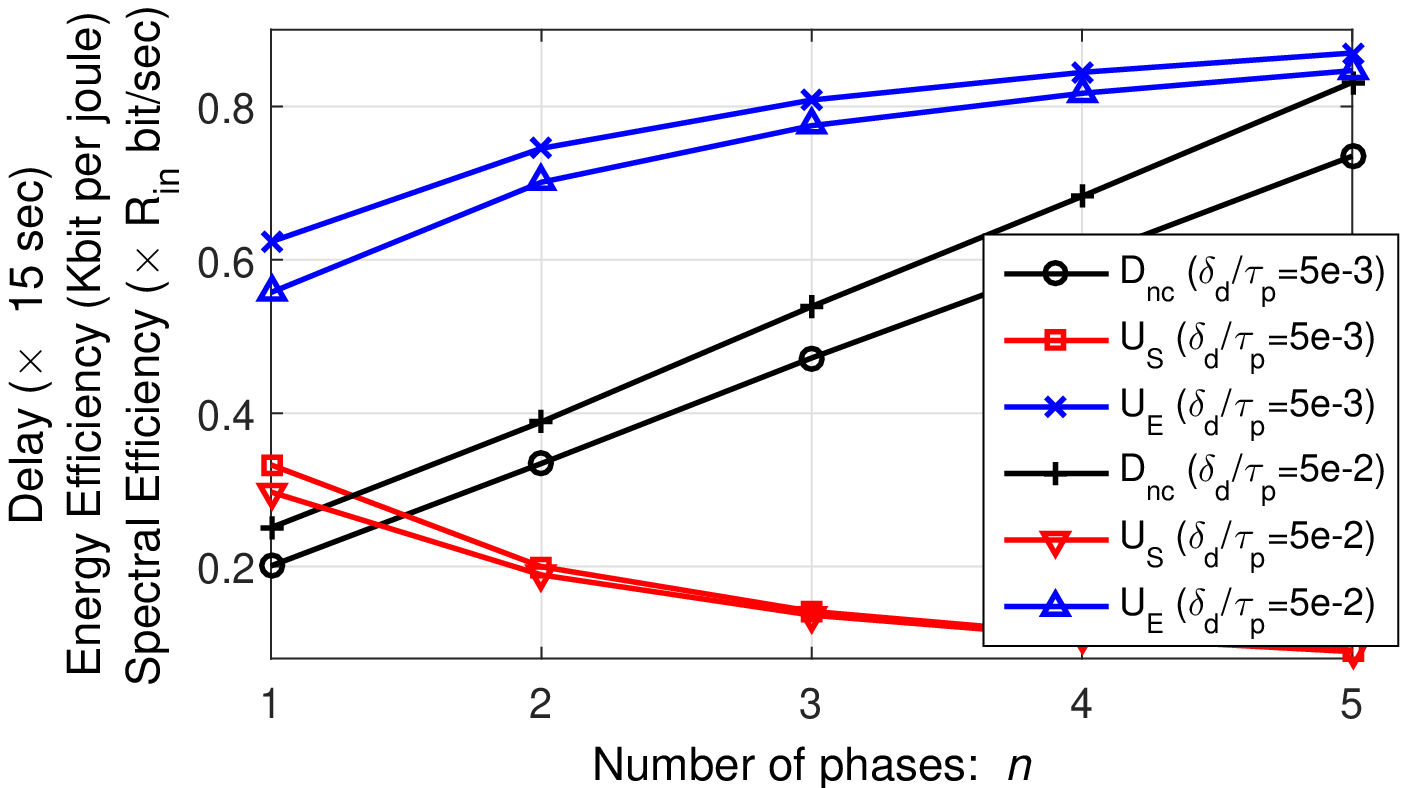}
\caption{ Energy  efficiency, delay, and spectral efficiency for the $n$-phase CSMA/CA. The parameters are the same as Fig. \ref{figt}.} \label{mph}
  \end{minipage}
~
  \begin{minipage}[b]{0.45\textwidth}
    \includegraphics[width=1.1\textwidth]{ 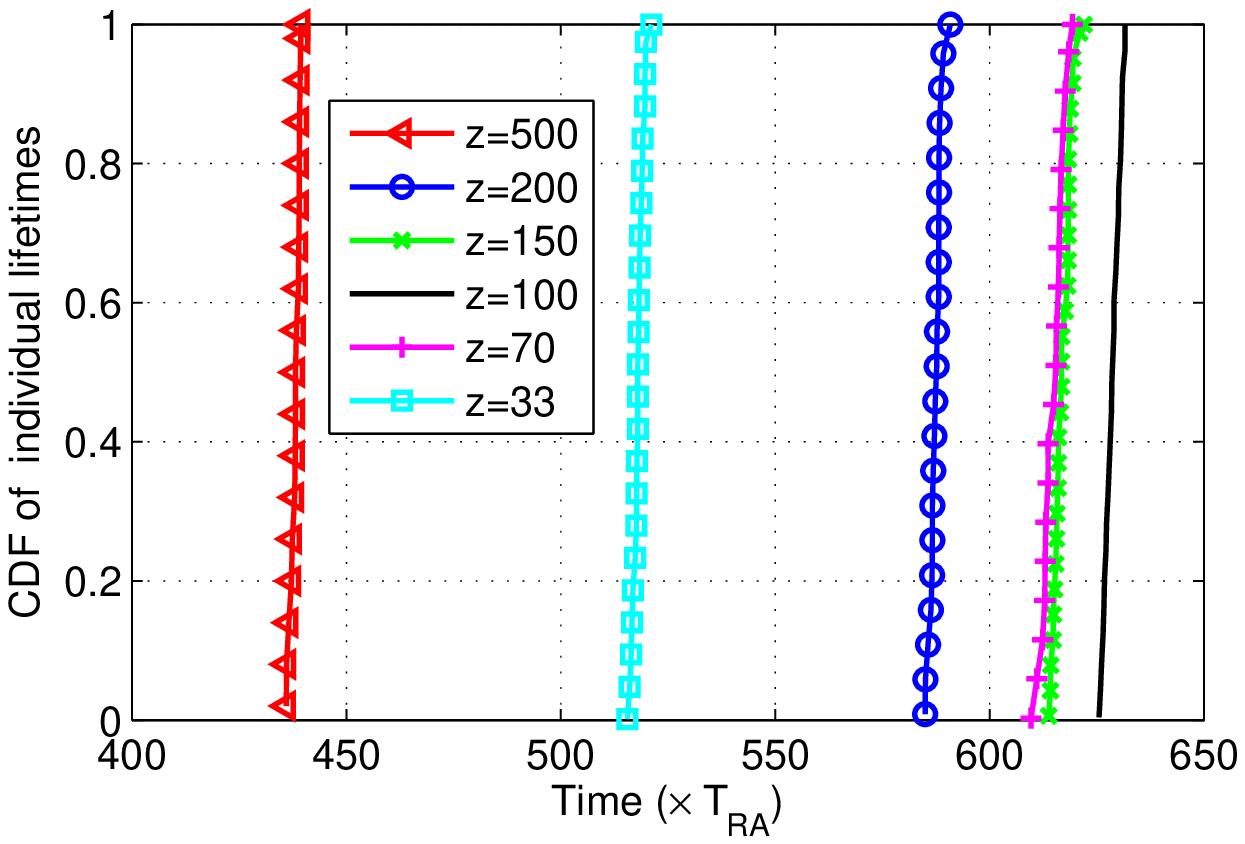}
\caption{The expected CDF of individual lifetimes has been depicted versus cluster-size.  } \label{life}
  \end{minipage}
\end{figure}

%


\begin{table}[t!]
\centering \caption{Simulation Parameters   }\label{sim3}
\begin{tabular}{p{7 cm}p{8 cm}}\\
\toprule[0.5mm]
{\it Parameters }&{\it Value}\\
\midrule[0.5mm]
Cell outer and inner radius &  $500$, 50 m\\
Pathloss, $\Omega_h(d)$ & $128.1+37.6\log (\frac{d}{1000})$\\
Pathloss, $\Omega_m(d)$ &  $38.5+20\log (d)$\\
Thermal noise power & $-204$ dBW/Hz \\
Number of devices & $5000$  \\
Available resources& 180 KHz $\times$ 2.4 sec per $T_{RA}$: 240 LTE frames\\
$T_{RA}$&1000 sec\\ 
$P_c, P_t^m, P_t^h$& 20 mW, 50 mW, 200 mW\\
$E_s^h$ & 1.5 mJ per $T_{RA}$\\
\hline
{\it Traffic parameters} &{ }\\
Packet arrival of each device, $r_g$&Poisson distributed. Average: 1  per 7 hours \\
Packet size & $5$ Kbytes\\
\hline
{\it cMAC parameters} &{ }\\
Communications protocol& Reservation through config. 0 of RACH \cite{book1}, communications through PUSCH\\
Number of preambles& 54 in even frames\\
\hline
{\it Intra-cluster parameters}&{ }\\
Communications protocol& $n$-phase CSMA/CA\\
Time for intra-cluster communications of all clusters& 1.4 sec (140 frames)\\
Time for intra-cluster communications of each cluster, $T_{intra}$& $ \min\{z,200\} $ msec\\
$\theta_b,\theta_f$ &$\frac{T_{intra}}{5n}$\\ 
$\delta_d$& 1 msec\\
\hline
{\it Inter-cluster   parameters} &{ }\\
Communications protocol& Reservation through PUCCH, communications through PUSCH \cite{sched}\\
$T_{inter}$& 1 sec (1000 PRBPs)\\
\bottomrule[0.5mm]
\end{tabular}
\end{table}


\section{Performance Evaluation}\label{sim}
In this section, we evaluate the system performance. To this end, the uplink transmission of 5000 machine nodes which are  randomly distributed according to a spatial Poisson point process in a single cell with one BS at the center is simulated using MATLAB. 


\subsection{Structure of the implemented MAC schemes}\label{e2}
The  implemented $E^2$-MAC follows the presented structure in Fig. \ref{Lframe}. In this figure, $T_{RA}$ shows the time interval between two consecutive resource allocations to the machine nodes. $E^2$-MAC benefits from the $n$-phase CSMA/CA for  communications inside the clusters. When the allocated phase for a group of CMs starts, each node which has data to transmit waits for a random  time window, which is exponentially distributed with mean $\theta_b$, and then, sends its packets.  For communications between CHs and the BS, CHs  reserve PUSCH resources in advance, e.g. using the physical uplink control channel  \cite{sched} or by persistent resource reservation \cite{reli}.  The detailed simulation parameters  can be found in Table \ref{sim3}.  From Table \ref{sim3}, one sees that the maximum number of allocated frames for intra-cluster communications of each cluster is 20, however, the total number of available frames for intra-cluster communications of all clusters is 140. Then,  the BS can allocate 7 orthogonal bunches of frames to 7 neighbor clusters in order to mitigate the inter-cluster interference.  As a benchmark,  performance of the $E^2$-MAC  is compared against a contention-based MAC (cMAC) protocol which is designed based on the configuration 0 of the  RACH of LTE \cite{book1}. In cMAC, 54 orthogonal preambles are available in the second subframes of even-numbered frames for resource reservation of machine nodes that have data to transmit. Also, data transmission of successful nodes in resource reservation at frame $i$ will be scheduled to be done in frames $i+1$ and $i+2$.  

\subsection{Analytical results}\label{anati}
To find the cluster size that maximizes the FED lifetime, we analyze the expected network lifetime for different cluster-size values. Using the proposed framework in section \ref{opsi} and the energy consumption expressions for CSMA/CA protocol in section \ref{bse}, one can rewrite the $L_c (d_h,z)$ expression in \eqref{FEDl} by inserting the following parameters:
\begin{align}
&F_{\mathcal M}(w,P_t^m,\Omega_m (\sqrt{\frac{z}{4\sigma}}),z)\simeq  \frac{p_{is} w}{r_g T_{RA}}\log(1+\frac{P_t^m}{N_0w\Gamma \Omega_m(\sqrt{\frac{z}{4\sigma}})}),\nonumber\\
&F_{\mathcal H}(w, P_t^h,\Omega_h (d_h),\frac{N_t}{z})=w\log(1+\frac{P_t^h}{N_0w\Gamma \Omega_h(d_h)}),\nonumber\\
&T_c=T_{RA}, E_s= r_g T_{RA}  P_c \theta_b, E_s^h=P_c T_{intra}+1.5 \text{mJ},\nonumber
\end{align}
where $r_g$ is the packet generation rate of each node, and $p_{is}$ has been derived in section \ref{bse}.
Fig. \ref{life} depicts the cumulative density function  of $L_c(d_h,z)$ for different  $z$ values when $d_h$ is the distance between a randomly chosen point in the cell and the BS. From this figure, one sees that $z=100$ outperforms the others and achieves the highest FED network lifetime.  Also, one sees that both having too many or too small number of clusters in the cell can degrade the network lifetime significantly. 
\begin{figure}
         \begin{subfigure}[b]{0.5\textwidth}
                \includegraphics[trim={0cm 0cm 0cm 0cm},clip,width=3.5in]{ 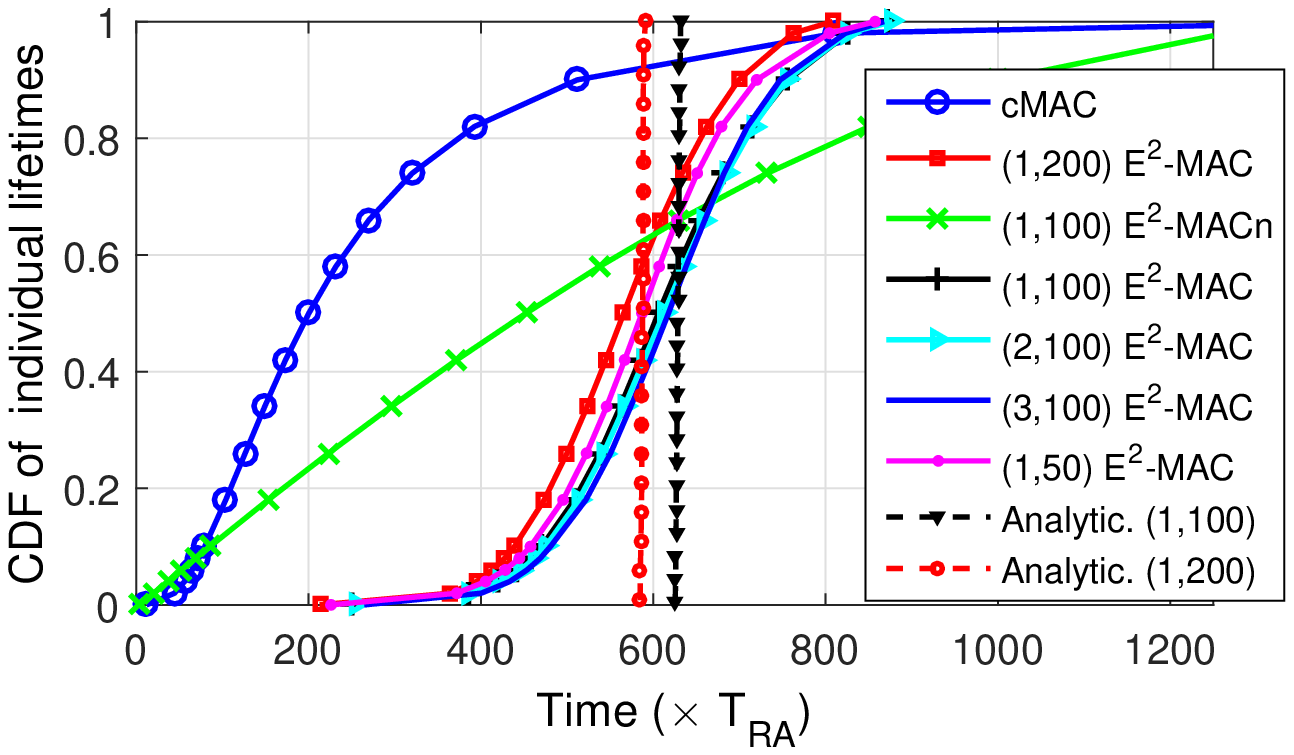}
                \caption{CDF of individual lifetimes}
                \label{pdf}
        \end{subfigure} 
        ~
         \begin{subfigure}[b]{0.5\textwidth}
                \includegraphics[trim={0cm 0cm 0cm 0cm},clip,width=3.5in]{ 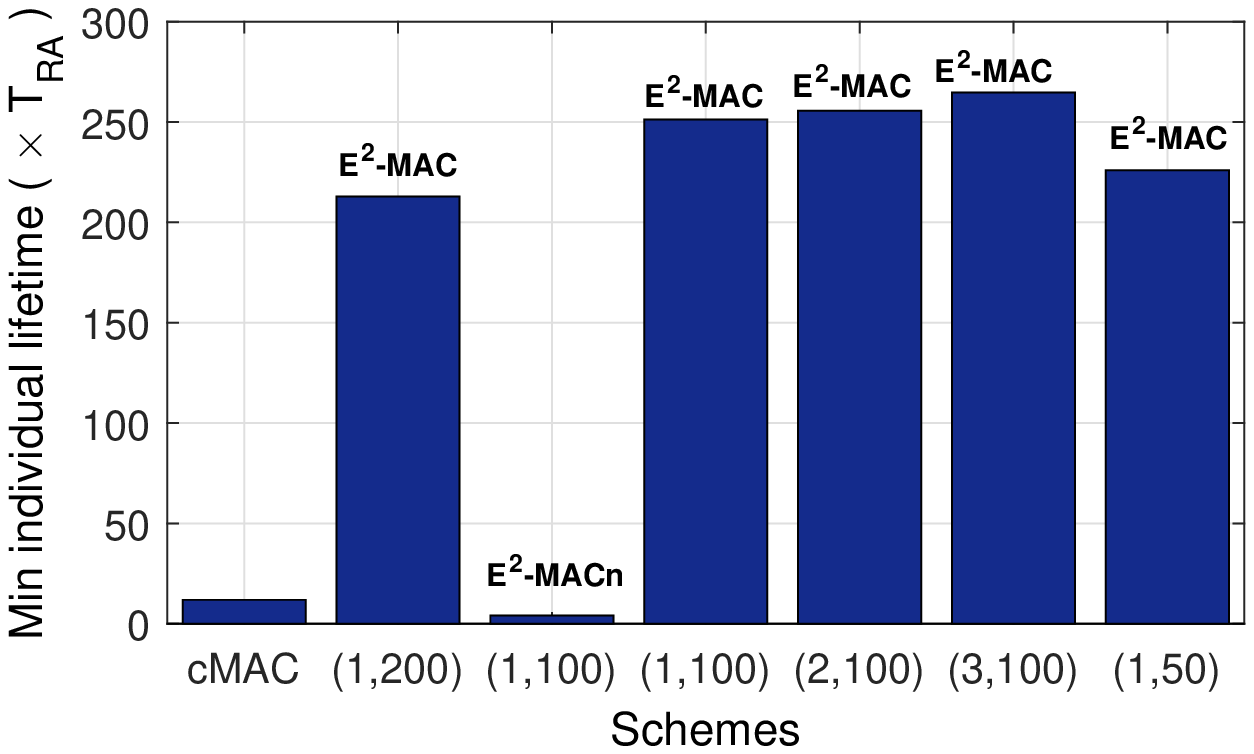}
                \caption{Detailed lifetime comparison}
                \label{pdfn}
        \end{subfigure} 
        \caption{Lifetime performance comparison of different MAC protocols}\label{fig_asl1}
\end{figure}

\begin{figure}
\centering
                \includegraphics[trim={0cm 0cm 1cm 0cm},clip,width=3.5in]{ 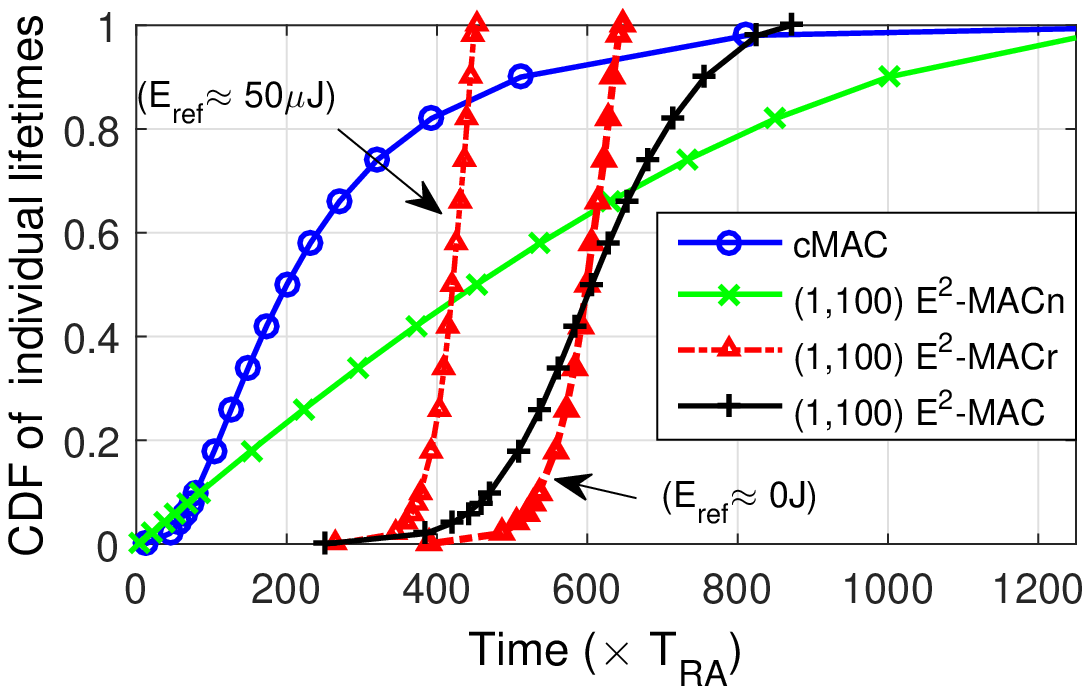}
                \caption{Lifetime performance of cluster-based MTC with cluster-reforming}
                \label{reft}
        \end{figure} 

\subsection{Simulation results}
In the following figures, ($x$,$y$)$E^2$-MAC refers to the $E^2$-MAC where $x$ is the number of phases for the $n$-phase CSMA/CA and
 $y$ is the average cluster size. Also, $E^2$-MACn refers to a version of the $E^2$-MAC in which CH reselection happens after death of each CH, i.e. the current set of CHs will remain in the CH mode until death.
 Fig. \ref{fig_asl1} compares lifetime performance of  the $E^2$-MAC with the lifetime-maximizing cluster-size, i.e. $z=100$, against the $E^2$-MAC with non-optimal cluster size and the cMAC.   First, Fig. \ref{pdf} represents the evolution of the individual battery lifetimes from the reference time at which all devices are fully charged until the last battery is depleted. One sees that using the cMAC, a great number of nodes die very early because of energy wastage in collisions and idle listening, and the remaining nodes  last for a longer time because of reduced contention for channel access. Furthermore, we see that using the $E^2$-MACn, the respective CDF curve has a mild slope because the first set of CHs drains out of energy very soon and  the last set of CHs lasts for a very long time.  Also, using  (1,100)$E^2$-MAC, where 100 is the lifetime maximizing cluster-size as derived in Fig. \ref{life}, one sees the CDF curve has a steeper slope which means almost all machine nodes  die in a limited time-window indicating replacement of their batteries can be done all at once. The semi-vertical curves in this figure  present the expected CDF of individual lifetimes as we derived from the analytical results in Fig. \ref{life}. One sees that the derived curves from the simulation results are centered on their expected values but the slopes of these curves are not as sharp as the slopes of the expected curves. In other words, we expect from the analytical results that all nodes die almost at the same time, but in simulations nodes  die in a  time-window. This difference is due to the fact that in our analytical model in \eqref{FEDl} we have assumed that all clusters have the same cluster-sizes, however, in simulations different clusters may have different numbers of CMs which can significantly impact the network lifetime. Also, our lifetime model in \eqref{FEDl} assumes that all CMs have the same lifetimes, however, in simulations the CMs will die sequentially which means the last node in a cluster will die approximately $zT_{RA}$ seconds later than the first node.    Finally, it is evident that the lifetime can be further improved by increasing the number of phases for the $n$-phase CSMA/CA, e.g. by using (3,100)$E^2$-MAC instead of (1,100)$E^2$-MAC.  
 
  The detailed FED network lifetime performance comparison of the proposed MAC schemes is presented in Fig. \ref{pdfn}.  In this figure, it is evident that the $E^2$-MACn achieves the worst FED network lifetime, because using this scheme the first set of selected CHs dies very early. On the other hand, this scheme achieves the longest individual lifetime, which makes it favorable in specific metering applications. 
  Also, it is evident that the (3,100)$E^2$-MAC achieves the best FED network lifetime performance. 

Fig. \ref{reft} evaluates lifetime performance of cluster-based M2M communications with cluster reforming. In this figure, $E^2$-MACr represents a version of $E^2$-MAC in which, after each CH re-selection machine nodes connect to the nearest CH, and hence, cluster-reforming may happen. As discussed in section \ref{refor}, cluster-reforming can prolong the network lifetime if the amount of saved energy in reforming the clusters is larger than the consumed energy per node in cluster-reforming procedure. On sees in Fig. \ref{reft} that when $E_{ref}\approx 0$, i.e. the consumed energy per device for cluster-reforming is negligible, the FED network lifetime of $E^2$-MACr is 55\% larger than the one of $E^2$-MAC. However, when $E_{ref}= 50 \mu$J, this improvement is only 5\%. Then, an efficient implementation of $E^2$-MACr can contribute in prolonging the network lifetime.

\begin{figure}
        \centering        
        \begin{subfigure}[b]{0.5\textwidth}
               \includegraphics[trim={0cm 0cm 0cm 0cm},clip,width=3.5in]{ 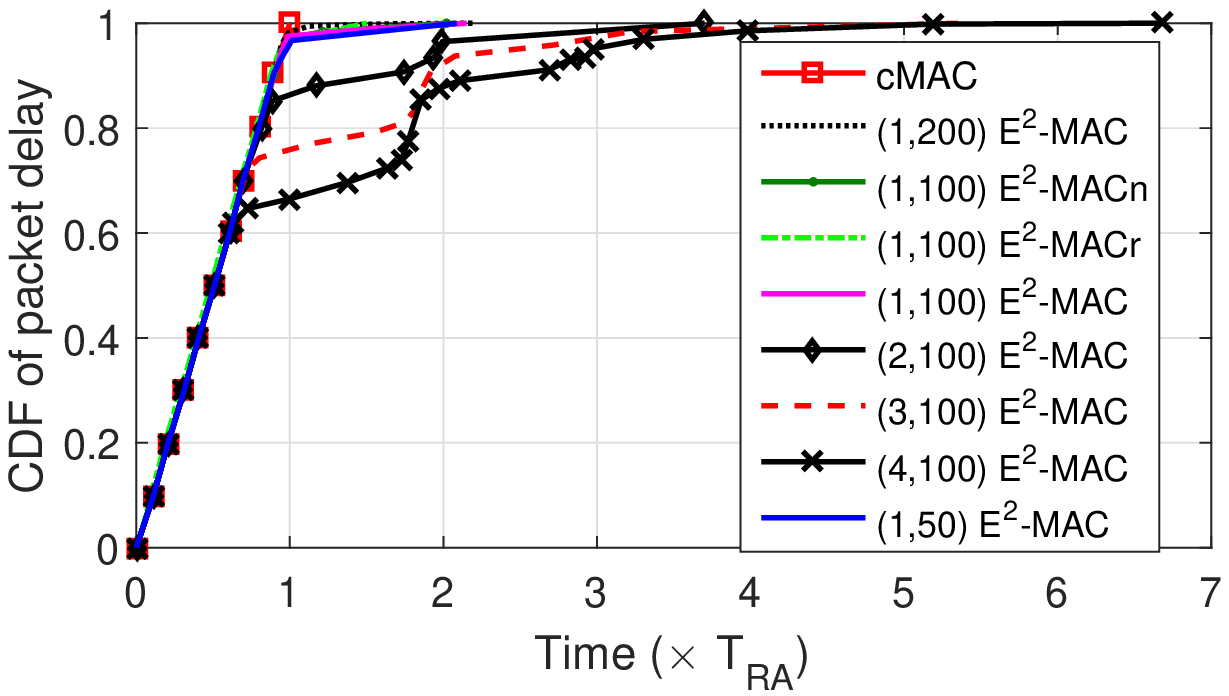}
               \caption{CDF of packet delay }
                \label{de}
        \end{subfigure}~
                \centering
         \begin{subfigure}[b]{0.5\textwidth}
                \includegraphics[trim={0cm 0cm 0cm 0cm},clip,width=3.5in]{ 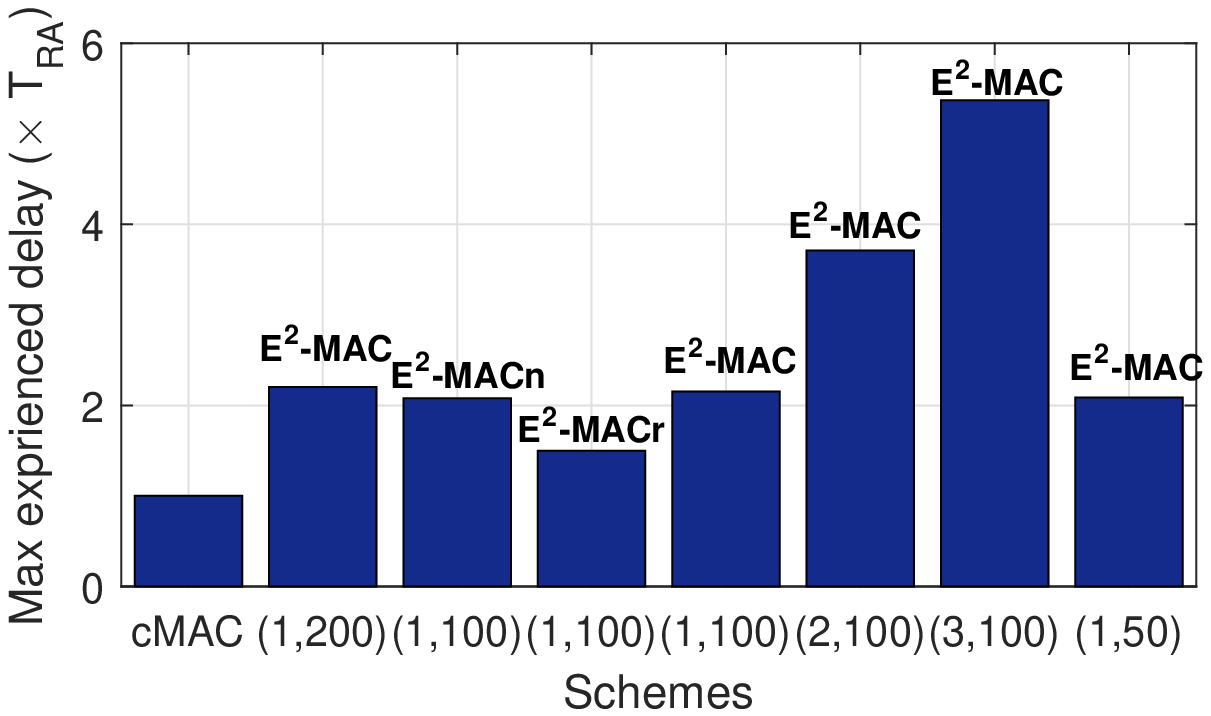}
                \caption{Detailed delay comparison}
                \label{den}
        \end{subfigure} 
        \caption{Delay performance comparison of different MAC protocols}\label{fig_asl2}
\end{figure}

 Fig. \ref{de} represents the CDF of packet delay for different MAC schemes. One sees that using $n$-phase CSMA/CA,  packet delay increases in the number of phases.   The detailed delay performance comparison  is presented in Fig. \ref{den}.  In this figure, we see that the maximum experienced delay by (n,100) $E^2$-MAC is  approximately $0.7 n$ higher than the (1,100) $E^2$-MAC scheme.  By comparing Fig. \ref{pdfn} and Fig. \ref{den}, one sees that the $n$-phase CSMA/CA offers a tunable tradeoff between energy efficiency and packet delay, because  both  lifetime and packet delay increase in the number of phases. Also, one sees that the maximum experienced delay in (1,100) $E^2$-MACr scheme is less than the one of (1,100) $E^2$-MAC. This is due to the fact that the average communications distance in the latter is shorter than the former, as discussed in section \ref{refor}.

From the lifetime and delay analyses in Fig. \ref{pdf}-Fig. \ref{fig_asl2}, one sees that the (1, $z^*$) $E^2$-MAC can significantly improve the FED network lifetime. Also, we see that further lifetime improvement is achievable at the cost of sacrificing the delay performance by utilizing the ($n, z^*$) $E^2$-MAC scheme, where $n>1$. Then, for M2M networks in which the performance/coverage is affected by losing some nodes, ($n, z$) $E^2$-MAC can be used, in which $n$ and $z$ are tuned based on the system parameters, delay budget, and available time/frequency resources. Furthermore, for M2M networks in which the correlation between gathered data by different nodes is high, and hence, the longest individual lifetime is defined as the network lifetime, the $E^2$-MACn achieves the best lifetime performance.

\section{Conclusion}
In this paper, we have proposed $E^2$-MAC to maximize network battery lifetime in massive M2M networks. Theoretical analyses are provided on the impact
of clustering, cluster size, and cluster-head selection 
on both individual lifetime of machine nodes and network lifetime. 
It is shown that there is a cluster size which maximizes the network lifetime and this cluster size is formulated as a function of system parameters. To further prolong the network lifetime, a decentralized cluster-head (re-)selection scheme  is also presented. Furthermore, by investigating the feasibility of clustering in different regions of the cell it is shown that clustering may not be a lifetime-aware scheme in some regions. Then, a general condition which must be satisfied by any feasible region is derived. Finally, a tunable delay-energy tradeoff for intra-cluster communications is  obtained by devising  an energy-efficient $n$-phase CSMA/CA scheme which can be tuned to provide a close-to-zero energy wastage  for cluster members. 


\section*{Acknowledgment}
The authors would like to thank X. Chen and P. Zhang for helpful investigation of feasibility of the project.

\ifCLASSOPTIONcaptionsoff
  \newpage
\fi

\bibliographystyle{IEEEtran}
\bibliography{bibl}

\end{document}